\documentclass[a4paper,aps,prd,10pt,preprintnumbers,showpacs,showkeys,onecolumn,superscriptaddress,nofootinbib,amsmath,amssymb]{revtex4-2}
\usepackage{amsthm,bm,mathtools}
\usepackage{tabularray}
\UseTblrLibrary{booktabs}
\usepackage[hidelinks]{hyperref}
\newtheorem{proposition}{Proposition}
\newtheorem{theorem}{Theorem}

\newcommand{\dd}{\mathrm{d}}
\newcommand{\ii}{\mathrm{i}}
\newcommand{\cU}{\cos(aU)}
\newcommand{\sU}{\sin(aU)}

\begin{document}

\title{Homogeneous Spectral Freedom and Hadamard Consistency\\
of Perturbative Green Functions}
\author{Bekir Can L\"utf\"uo\u{g}lu}
\altaffiliation{ORCID: \href{https://orcid.org/0000-0001-6467-5005}{0000-0001-6467-5005}}
\email{bekir.lutfuoglu@uhk.cz}
\affiliation{Department of Physics, Faculty of Science, University of Hradec Kr\'alov\'e, Rokitansk\'eho 62/26, 500 03 Hradec Kr\'alov\'e, Czech Republic}

\begin{abstract}
Perturbative mode equations often admit formally different particular solutions whose spectral sums appear to have different short-distance singularities.  We formulate a general diagnostic which separates this representational freedom from a genuine change of quantum state: for a fixed normally hyperbolic operator, two Hadamard two-point functions can differ only by a smooth homogeneous bisolution.  Thus, a nonsmooth term generated by rearranging a perturbative source must be traced either to a changed operator or to inadmissible homogeneous spectral data.  We apply the criterion to two impulsive plane-wave backgrounds.  Starting directly from the metrics, we recompute the wave operators, modes, world functions, and Van Vleck determinants.  Exact-mode matching and the geometrical construction agree through second order for the first profile.  In the second, correcting the metric changes the leading perturbation from quadratic to linear order in the wave parameter and eliminates inverse transverse-momentum poles, while the corrected metric retains an exact aligned null-radiation Einstein source.  In both examples the singular Green function is precisely the perturbative expansion of the Hadamard parametrix; the additional logarithms arise from momentum-singular homogeneous data.  For the exactly soluble first profile, the standard flat-normalized plane-wave prescription gives a vanishing renormalized stress tensor away from the impulsive surface.  Smooth-front regularizations of both profiles reproduce, at the calculated orders and away from the characteristic zero mode, the singular structure of the matched impulsive propagators without generating inverse transverse-momentum poles.
\end{abstract}
\maketitle

\section{Introduction} \label{sec:introduction}

The short-distance singularity of a quantum two-point function in curved spacetime is highly constrained.  For a physically admissible Hadamard state, the singular coefficients are determined locally by the metric and the field equation; different Hadamard states can differ only by a smooth bisolution \cite{Fulling1978,KayWald1991,Radzikowski1996,Brunetti1996,HollandsWald2001}. This statement underlies covariant point-splitting renormalization \cite{DeWittBrehme1960,Christensen1976,DecaniniFolacci2006} and makes the singularity structure itself a sensitive test of perturbative calculations.  If two constructions on the same background produce inequivalent singular terms, they cannot both describe Hadamard states with the same boundary or matching prescription.

Plane-wave spacetimes provide an unusually sharp setting for this test.  They are nontrivial gravitational backgrounds for which propagation is often tractable, while their null symmetry strongly restricts quantum effects. For plane sandwich waves, Gibbons showed that the time-dependent geometry does not produce particles and that the renormalized stress tensor vanishes under the standard vacuum prescription \cite{Gibbons1975,Deser1975}. Related analyses of zero-mass waves in curved gravitational backgrounds appear in Ref.~\cite{Chrzanowski1976}.  The propagation of massless fields is also closely tied to the special tail structure of plane waves \cite{Harte2013}.  These facts do not make every intermediate Green function trivial, but they provide stringent consistency checks on its state dependence and ultraviolet form.

Impulsive waves add a separate subtlety.  Their curvature is concentrated on a null hypersurface, so products of distributions and operations performed directly at the wavefront require care.  Smooth-profile regularizations and low-regularity geometric methods provide a controlled interpretation of this limit \cite{Steinbauer1998,KunzingerSteinbauer1999,SteinbauerVickers2006, PodolskySteinbauer2022,Samann2024}.  Exact impulsive solutions on constant-curvature backgrounds provide a complementary geometrical framework \cite{PodolskyGriffiths1999}.  Away from the wavefront, however, the metric is smooth and the ordinary local Hadamard analysis applies without modification.  It is therefore essential not to confuse a distributional matching question at the impulse with an additional short-distance singularity in a smooth region. Recent exact constructions of Wightman functions on impulsive plane-wave backgrounds likewise organize their local singularities through the world function and Van Vleck determinant and show that separate changes occur at conjugate planes \cite{Cho2023}.  Our analysis is deliberately restricted to normal neighbourhoods before such caustics.

The 1999 calculation of Horta\c{c}su and L\"utf\"uo\u{g}lu investigated two members of Hogan's impulsive plane-wave family and observed that different splittings of an inhomogeneous perturbative mode equation could generate or remove logarithmic terms \cite{HortacsuLutfuoglu1999,Hogan1999}.  The work correctly identified these terms as spurious at the level of the final renormalized stress tensor.  It did not, however, fully separate the choice of a particular solution from the homogeneous spectral data, and the second example inherited inconsistencies already present in the specialized metric. This leaves a sharper question: are the logarithms produced by the geometry, or do they represent a change of state hidden inside the perturbative solution?

We answer this question by rebuilding both examples from the metric.  For the first profile we compare exact modes, characteristic solutions, and the world-function construction through second order.  For the second we first correct the metric specialization and then derive its leading mode and Green function corrections.  We show that source redistribution alone cannot alter the field: the difference resides in freely chosen homogeneous functions. Exact incoming-mode matching and Hadamard admissibility fix that freedom and leave only the geometrical singularity.  We then examine the stress tensor, state the result as Proposition~\ref{prop:main}, and give an equation-by-equation comparison with the earlier calculation.

The paper is organized as follows.  Section~\ref{sec:general-diagnostic} formulates the general Hadamard diagnostic. Sections~\ref{sec:first-profile} and \ref{sec:first-green-functions} develop the exact and perturbative mode analysis of the first profile and compare its matched Green function with the geometrical Hadamard expansion.  Section~\ref{sec:impulsive-limit} treats matching across the impulse and the smooth-front limit, while Sec.~\ref{sec:stress-tensor} discusses Hadamard subtraction and the renormalized stress tensor.  The corrected second profile is analysed in Sec.~\ref{sec:second-profile}.  Section~\ref{sec:principal-discussion} collects the principal result, its implications, and the conclusions.  The equation-by-equation audit of the 1999 calculation is provided in Appendix~\ref{sec:comparison}.

\section{A general Hadamard diagnostic} \label{sec:general-diagnostic}

We consider a real massless scalar field with curvature coupling $\xi$, whose smooth-background field operator is
\begin{equation}
 P=\Box+\xi R_{\rm sc}.  \label{eq:scalar-field-operator}
\end{equation}
Both plane-wave profiles studied below have $R_{\rm sc}=0$, so their mode equation is $\Box\phi=0$ for every $\xi$; the coupling remains relevant only to the point-split stress-tensor operator.  We use signature $(+---)$ and define the Feynman two-point function by
\begin{equation}
 G_F(X,X')=\langle T\phi(X)\phi(X')\rangle, \qquad  P_XG_F(X,X')=-\ii\,\delta_g(X,X'). \label{eq:Feynman-convention}
\end{equation}
With $D_0=2\sigma_0$, where $\sigma_0$ is the flat-space world function, the flat massless kernel is therefore
\begin{equation}
 G_F^{(0)}(X,X')=-\frac{1}{4\pi^2}\frac{1}{D_0-\ii0}, \qquad  \mathcal N\equiv-\frac{1}{4\pi^2}. \label{eq:Feynman-normalization}
\end{equation}
Here and below $-\ii0$ is shorthand for the usual covariant Feynman boundary value in a convex normal neighbourhood.

Let $P_\epsilon$ be a smooth one-parameter family of normally hyperbolic operators on a globally hyperbolic region $\mathcal O$, with formal expansion
\begin{equation}
 P_\epsilon=P_0+\epsilon P_1+\epsilon^2P_2+\cdots .  \label{eq:general-operator-expansion}
\end{equation}
Suppose that a perturbative two-point function is constructed either from modes or directly from a Green-function equation.  At order $n$, the inhomogeneous equation fixes a particular solution but permits addition of a bisolution of $P_0$.  In a mode representation this freedom appears as arbitrary functions of characteristic invariants or as unfixed spectral coefficients.  The smooth-difference theorem for Hadamard states is standard; the point recorded below is its order-by-order consequence for perturbative source splittings.  The new content of this work lies in turning that consequence into an operator--matching--parametrix diagnostic and applying it to the two plane-wave problems.

\begin{theorem}[Hadamard diagnostic for perturbative splittings] \label{thm:general-diagnostic}
Let $G_\epsilon^+$ and $\widetilde G_\epsilon^+$ be perturbative Wightman two-point functions for the same operator $P_\epsilon$ on $\mathcal O$, constructed through order $N$, and suppose that they obey the same canonical commutator.  Assume also that their Hadamard remainders admit $C^N$ expansions in $\epsilon$ with coefficients in $C^\infty(\mathcal O\times\mathcal O)$.  Their first nonvanishing difference at order $n\leq N$ is a homogeneous bisolution of $P_0$.  If both constructions are Hadamard, this difference is smooth.  The analogous statement holds for Feynman kernels provided that both solve the same delta-source equation and implement the same boundary prescription.  Hence a source splitting which produces a new nonsmooth term without changing $P_\epsilon$ necessarily changes the homogeneous state data and cannot represent a second Hadamard realization of the same matched state.
\end{theorem}

\begin{proof}
Write
\begin{equation}
 G_\epsilon^+=\sum_{j=0}^{N}\epsilon^jG_j^++O(\epsilon^{N+1}), \qquad
 \widetilde G_\epsilon^+=\sum_{j=0}^{N}\epsilon^j\widetilde G_j^+ +O(\epsilon^{N+1}). \label{eq:general-G-expansion}
\end{equation}
If the coefficients agree below order $n$, subtraction of the two order-$n$ Green-function equations gives
\begin{equation}
 P_0(G_n^+-\widetilde G_n^+)=0  \label{eq:general-homogeneous-difference}
\end{equation}
in each argument away from the diagonal; equality of the commutators removes the antisymmetric fundamental-solution part.  For Feynman kernels, the same subtraction instead cancels the common delta source, while equality of the boundary prescription fixes the fundamental-solution part.  The local Hadamard coefficients are fixed by $P_\epsilon$ and the geometry.  The difference of any two Hadamard two-point functions for that operator is therefore a smooth bisolution \cite{Fulling1978,KayWald1991,Radzikowski1996,Brunetti1996}. The assumed regularity in $\epsilon$ permits this conclusion to be taken coefficient by coefficient. A nonsmooth difference cannot be attributed to an alternative choice of particular solution within the same Hadamard state.  It must enter through different homogeneous data (or, if the operators were not in fact equal, through a changed field equation).
\end{proof}

For the applications below, the operator is first derived from the metric, the complete homogeneous solution is retained when the source is rearranged, and the remaining freedom is fixed by common matching data.  The resulting singularity can then be compared with the geometrical Hadamard parametrix. The first profile admits an exact mode solution.  In the second, the same comparison is made after correcting the background metric.

\section{First profile: geometry and mode solutions} \label{sec:first-profile}

\subsection{Geometry of the plane-wave background} \label{sec:first-geometry}

We begin with the impulsive gravitational plane-wave metric considered by Hogan \cite{Hogan1999},
\begin{equation}
 \dd s^2=-2\cos^2(aU)  \left|\dd z+\Theta(U)\frac{\tan(aU)}{a}\,  \overline{H}(\bar z)\dd\bar z\right|^2+2\dd U\dd V, \label{eq:complex-metric}
\end{equation}
where
\begin{equation}
 z=\frac{x+\ii y}{\sqrt{2}}, \qquad  \bar z=\frac{x-\ii y}{\sqrt{2}}.  \label{eq:complex-coordinates}
\end{equation}
For the first profile we take $G(z)=az$, so that $H(z)=a$.  In the $U>0$ region, Eq.~\eqref{eq:complex-metric} becomes
\begin{align}
 \dd s^2={}&2\dd U\dd V -\bigl[\cos(aU)+\sin(aU)\bigr]^2\dd x^2 -\bigl[\cos(aU)-\sin(aU)\bigr]^2\dd y^2. \label{eq:real-metric}
\end{align}
Thus, in coordinates $X^\mu=(U,V,x,y)$, the nonzero covariant components are
\begin{equation}
 g_{UV}=g_{VU}=1,\qquad g_{xx}=-(\cU+\sU)^2, \qquad  g_{yy}=-(\cU-\sU)^2.  \label{eq:metric-components}
\end{equation}
The inverse components are
\begin{equation}
 g^{UV}=g^{VU}=1,\qquad  g^{xx}=-\frac{1}{(\cU+\sU)^2},\qquad  g^{yy}=-\frac{1}{(\cU-\sU)^2},  \label{eq:inverse-components}
\end{equation}
and
\begin{equation}
 \sqrt{|g|}=\cos(2aU).  \label{eq:determinant}
\end{equation}
With the curvature convention
$R^\rho{}_{\sigma\mu\nu}
=\partial_\mu\Gamma^\rho_{\nu\sigma}
-\partial_\nu\Gamma^\rho_{\mu\sigma}
+\Gamma^\rho_{\mu\lambda}\Gamma^\lambda_{\nu\sigma}
-\Gamma^\rho_{\nu\lambda}\Gamma^\lambda_{\mu\sigma}$,
the only nonzero Ricci component in the post-front region is
\begin{equation}
 R_{UU}=2a^2, \qquad R_{\rm sc}=0.
 \label{eq:first-exact-Ricci}
\end{equation}
Thus the first profile carries an aligned null-radiation source.
Within a coordinate patch in which $\cos(2aU)>0$, the scalar wave operator
\begin{equation}
 \Box\Phi=\frac{1}{\sqrt{|g|}}\partial_\mu  \left(\sqrt{|g|}g^{\mu\nu}\partial_\nu\Phi\right),  \label{eq:box-definition}
\end{equation}
takes the form
\begin{align}
 \Box={}&2\partial_U\partial_V-2a\tan(2aU)\partial_V
 +\frac{\sin(2aU)-1}{\cos^2(2aU)}\partial_x^2 -\frac{\sin(2aU)+1}{\cos^2(2aU)}\partial_y^2.
 \label{eq:exact-box}
\end{align}
In particular,
\begin{equation}
 \Box=L_0+aL_1+O(a^2),
 \label{eq:operator-expansion}
\end{equation}
where
\begin{equation}
 L_0=2\partial_U\partial_V-\partial_x^2-\partial_y^2,
 \qquad
 L_1=2U\left(\partial_x^2-\partial_y^2\right).
 \label{eq:L0L1}
\end{equation}

\subsection{Exact separated modes}
\label{sec:exact-modes}

Consider a massless scalar mode
\begin{equation}
 \Phi=N(U)\exp\left[\ii RV+\ii k_1x+\ii k_2y+\ii Q(U)\right].
 \label{eq:mode-ansatz}
\end{equation}
Substitution into $\Box\Phi=0$ shows that the first-derivative term is
removed by
\begin{equation}
 N(U)=\frac{1}{\sqrt{\cos(2aU)}}.
 \label{eq:normalization}
\end{equation}
The remaining phase equation is
\begin{equation}
 Q'(U)=\frac{k_1^2+k_2^2-(k_1^2-k_2^2)\sin(2aU)}
 {2R\cos^2(2aU)},
 \label{eq:phase-equation}
\end{equation}
and hence
\begin{equation}
 Q(U)=\frac{1}{4aR}\left[
 (k_1^2+k_2^2)\tan(2aU)
 -\frac{k_1^2-k_2^2}{\cos(2aU)}\right]+Q_0.
 \label{eq:exact-phase}
\end{equation}
Equivalently, after absorbing $Q_0$ into the overall normalization,
\begin{align}
 \Phi={}&\frac{e^{\ii RV+\ii k_1x+\ii k_2y}}
 {\sqrt{\cos(2aU)}}\exp\left[
 \frac{k_1^2-k_2^2}{4\ii Ra\cos(2aU)}
 -\frac{k_1^2+k_2^2}{4\ii Ra}\tan(2aU)
 \right].
 \label{eq:exact-mode}
\end{align}

\subsection{Perturbative mode equation}
\label{sec:perturbative-mode-equation}

We now consider the off-shell eigenfunctions $\phi$ used in the resolvent representation, distinguishing them from the exact on-shell modes $\Phi$ above. Let
\begin{equation}
 \phi=\phi^{(0)}+a\phi^{(1)}+O(a^2),\qquad
 \lambda=\lambda^{(0)}+a\lambda^{(1)}+O(a^2),
 \label{eq:eigen-expansion}
\end{equation}
with
\begin{equation}
 \phi^{(0)}=\frac{1}{(2\pi)^2\sqrt{2|R|}}
 e^{\ii[(K/2R)U+RV+k_1x+k_2y]},
 \qquad
 \lambda^{(0)}=K-k_1^2-k_2^2.
 \label{eq:zeroth-mode}
\end{equation}
Here the four independent spectral labels $(K,R,k_1,k_2)$ belong to the
off-shell eigenfunction representation used in the resolvent.  Its
$(2\pi)^{-2}$ convention should not be confused with the three-label
Klein--Gordon normalization of the on-shell modes introduced below.
For $\lambda^{(1)}=0$, the first-order equation is
\begin{equation}
 \left(L_0-\lambda^{(0)}\right)\phi^{(1)}+L_1\phi^{(0)}=0.
 \label{eq:first-order-master}
\end{equation}
Using
\begin{equation}
 \phi^{(1)}=\phi^{(0)}\left(g(U,z)+h(U,\bar z)\right),
 \label{eq:first-order-ansatz}
\end{equation}
we obtain, after division by $2\phi^{(0)}$,
\begin{align}
 &\ii R(\partial_Ug+\partial_Uh)
 -\frac{\ii}{\sqrt2}(k_1+\ii k_2)\partial_zg
 -\frac{\ii}{\sqrt2}(k_1-\ii k_2)\partial_{\bar z}h -(k_1^2-k_2^2)U=0.
 \label{eq:correct-first-order-PDE}
\end{align}
Equation~\eqref{eq:correct-first-order-PDE}, rather than a chosen split of
it, will be the starting point for classifying admissible particular and
homogeneous solutions.

\subsection{Characteristic solution and homogeneous freedom}
\label{sec:characteristic-solution}

It is useful to introduce
\begin{equation}
 q=k_1+\ii k_2,\qquad \bar q=k_1-\ii k_2,
 \qquad d=k_1^2-k_2^2.
 \label{eq:q-d-definitions}
\end{equation}
The source in Eq.~\eqref{eq:correct-first-order-PDE} may be distributed
between the holomorphic and antiholomorphic equations by means of an
otherwise arbitrary function $J(U)$:
\begin{align}
 \ii R\partial_Ug-\frac{\ii q}{\sqrt2}\partial_zg&=J(U),
 \label{eq:g-split-general}\\
 \ii R\partial_Uh-\frac{\ii\bar q}{\sqrt2}\partial_{\bar z}h
 &=dU-J(U).
 \label{eq:h-split-general}
\end{align}
Adding Eqs.~\eqref{eq:g-split-general} and
\eqref{eq:h-split-general} recovers the same master equation for every
choice of $J(U)$.

The characteristic invariants are
\begin{equation}
 \zeta=z+\frac{q}{\sqrt2R}U,
 \qquad
 \bar\zeta=\bar z+\frac{\bar q}{\sqrt2R}U.
 \label{eq:characteristic-invariants}
\end{equation}
Consequently, the general separated solution is
\begin{align}
 g(U,z)&=F(\zeta)-\frac{\ii}{R}\int^U J(s)\,\dd s,
 \label{eq:g-general}\\
 h(U,\bar z)&=H(\bar\zeta)
 -\frac{\ii}{R}\int^U[d\,s-J(s)]\,\dd s,
 \label{eq:h-general}
\end{align}
where $F$ and $H$ are arbitrary one-variable functions.  Their sum is
therefore
\begin{equation}
 g+h=F(\zeta)+H(\bar\zeta)-\frac{\ii d}{2R}U^2+C.
 \label{eq:sum-general}
\end{equation}
The result is independent of the function $J(U)$.  In particular, for a
constant splitting parameter $\eta$ one may take
\begin{equation}
 J(U)=\eta dU,
 \label{eq:eta-splitting}
\end{equation}
which gives
\begin{align}
 g&=F(\zeta)-\frac{\ii\eta d}{2R}U^2,
 &
 h&=H(\bar\zeta)-\frac{\ii(1-\eta)d}{2R}U^2.
 \label{eq:eta-solutions}
\end{align}
Thus a genuine redistribution of the same inhomogeneous source cannot by
itself modify the perturbative mode.  Any difference between two such
solutions is contained entirely in the homogeneous functions $F$ and
$H$ (and in an irrelevant constant).

\subsubsection{Matching to the exact mode}
\label{sec:exact-mode-matching}

For an on-shell zeroth-order mode, $K=k_1^2+k_2^2$.  Expanding the exact
phase in Eq.~\eqref{eq:exact-phase} at small $a$ gives
\begin{equation}
 Q(U)=Q_{\mathrm{const}}(a)
 +\frac{k_1^2+k_2^2}{2R}U
 -a\frac{k_1^2-k_2^2}{2R}U^2+O(a^2),
 \label{eq:exact-phase-expansion}
\end{equation}
where the $U$-independent, singular phase $Q_{\mathrm{const}}(a)$ can be
absorbed into the mode normalization.  Moreover,
$[\cos(2aU)]^{-1/2}=1+O(a^2)$.  The exact solution therefore fixes
\begin{equation}
 \frac{\phi^{(1)}}{\phi^{(0)}}
 =-\frac{\ii(k_1^2-k_2^2)}{2R}U^2,
 \label{eq:exact-first-order-ratio}
\end{equation}
in precise agreement with the particular part of
Eq.~\eqref{eq:sum-general}.  Matching the perturbative basis to the exact
basis fixes the nonconstant homogeneous contribution to
\begin{equation}
 F(\zeta)+H(\bar\zeta)=0
 \label{eq:exact-matching-condition}
\end{equation}
up to a mode-normalization constant.

This calculation identifies the origin of possible additional
singularities.  They are not generated by redistributing the same source
between two valid equations: that redistribution cancels identically in
$g+h$.  They can arise only if extra mode-dependent homogeneous functions
are admitted.  Such terms change the spectral basis (or, equivalently,
the boundary/state prescription) and must be fixed before constructing a
Green function.

\section{First-profile Green functions and Hadamard structure}
\label{sec:first-green-functions}

\subsection{Matched Green function through first order}
\label{sec:matched-green-first-order}

We first consider two points $X=(U,V,x,y)$ and
$X'=(U',V',x',y')$ lying in the $U>0$ region.  The metric
\eqref{eq:real-metric} is of Rosen plane-wave form,
\begin{equation}
 \dd s^2=2\dd U\dd V-A_x^2(U)\dd x^2-A_y^2(U)\dd y^2,
 \label{eq:rosen-form}
\end{equation}
where
\begin{equation}
 A_x(U)=\cos(aU)+\sin(aU),\qquad
 A_y(U)=\cos(aU)-\sin(aU).
 \label{eq:rosen-scale-factors}
\end{equation}
Define
\begin{equation}
 I_i(U,U')=\int_{U'}^U\frac{\dd s}{A_i^2(s)},
 \qquad i=x,y.
 \label{eq:Ii-definition}
\end{equation}
Solving the transverse geodesic equations and evaluating the geodesic
action gives Synge's world function; the associated local bitensor and
Van Vleck machinery is reviewed in
Refs.~\cite{DeWittBrehme1960,PoissonPoundVega2011}:
\begin{equation}
 \sigma(X,X')=(U-U')(V-V')
 -\frac{U-U'}{2}\left[
 \frac{(x-x')^2}{I_x(U,U')}
 +\frac{(y-y')^2}{I_y(U,U')}
 \right].
 \label{eq:world-function-rosen}
\end{equation}
Because the transverse metric carries a minus sign, it is convenient in
Eq.~\eqref{eq:world-function-rosen} to regard $I_i$ as the positive
integrals defined in Eq.~\eqref{eq:Ii-definition}; the displayed minus
sign then reproduces the flat-space interval.

For the present scale factors the integrals are elementary:
\begin{align}
 I_x(U,U')&=\frac{1}{a}\left[
 \frac{1}{1+\tan(aU')}-\frac{1}{1+\tan(aU)}\right],
 \label{eq:Ix-exact}\\
 I_y(U,U')&=\frac{1}{a}\left[
 \frac{1}{1-\tan(aU)}-\frac{1}{1-\tan(aU')}\right].
 \label{eq:Iy-exact}
\end{align}
Their small-$a$ expansions are
\begin{align}
 I_x&=(U-U')\left[1-a(U+U')\right]+O(a^2),
 \label{eq:Ix-expansion}\\
 I_y&=(U-U')\left[1+a(U+U')\right]+O(a^2).
 \label{eq:Iy-expansion}
\end{align}
It follows that
\begin{equation}
 2\sigma=D_0-aS+O(a^2),
 \label{eq:sigma-first-order}
\end{equation}
where
\begin{align}
 D_0={}&2(U-U')(V-V')-(x-x')^2-(y-y')^2,
 \label{eq:D0-definition}\\
 S={}&(U+U')\left[(x-x')^2-(y-y')^2\right].
 \label{eq:S-definition}
\end{align}

The Ricci tensor of this background begins at order $a^2$.  Accordingly,
the local Van Vleck factor satisfies
\begin{equation}
 \Delta^{1/2}(X,X')=1+O(a^2)
 \label{eq:van-vleck-first-order}
\end{equation}
in a normal neighbourhood.  With the convention and fixed normalization in
Eq.~\eqref{eq:Feynman-normalization}, the singular part of the matched
massless Feynman Green function is therefore
\begin{align}
 G_F(X,X')&=\mathcal N\frac{1}{D_0-aS-\ii0}+O(a^2)
 \notag\\
 &=\mathcal N\left[
 \frac{1}{D_0-\ii0}
 +a\frac{S}{(D_0-\ii0)^2}
 \right]+O(a^2).
 \label{eq:GF-first-order}
\end{align}
No logarithmic enhancement is present at this order.  The correction has
the standard differentiated-Hadamard form $S/(D_0-\ii0)^2$ and is exactly
what is obtained from the matched mode correction
\eqref{eq:exact-first-order-ratio} in the spectral representation.

Equation~\eqref{eq:GF-first-order} also makes the geometrical status of
the first-order result transparent.  Since the curvature starts at
$O(a^2)$, the $O(a)$ modification is a coordinate deformation of the
flat-space singularity rather than a new local vacuum-polarization
effect.

\subsection{Second-order Hadamard structure}
\label{sec:second-order-hadamard}

The exact transverse integrals also provide the second-order world
function without solving a new differential equation.  Writing
\begin{equation}
 p=U+U',\qquad r=U^2+UU'+U'^2,
 \label{eq:p-r-definitions}
\end{equation}
we find
\begin{align}
 I_x&=(U-U')\left[1-ap+\frac{4}{3}a^2r\right]+O(a^3),
 \label{eq:Ix-second-order}\\
 I_y&=(U-U')\left[1+ap+\frac{4}{3}a^2r\right]+O(a^3).
 \label{eq:Iy-second-order}
\end{align}
The combinations entering the world function reduce to
\begin{align}
 \frac{U-U'}{I_x}&=1+ap-\frac{a^2}{3}(U-U')^2+O(a^3),
 \label{eq:inverse-Ix}\\
 \frac{U-U'}{I_y}&=1-ap-\frac{a^2}{3}(U-U')^2+O(a^3).
 \label{eq:inverse-Iy}
\end{align}
Consequently,
\begin{equation}
 2\sigma=D_0-aS+a^2T+O(a^3),
 \label{eq:sigma-second-order}
\end{equation}
where $D_0$ and $S$ are given in Eqs.~\eqref{eq:D0-definition} and
\eqref{eq:S-definition}, and
\begin{equation}
 T=\frac{(U-U')^2}{3}\left[(x-x')^2+(y-y')^2\right].
 \label{eq:T-definition}
\end{equation}

For a diagonal Rosen wave with two transverse directions, the Van Vleck
determinant can be written locally as
\begin{equation}
 \Delta(X,X')=
 \frac{(U-U')^2}
 {A_x(U)A_y(U)A_x(U')A_y(U')I_x(U,U')I_y(U,U')}.
 \label{eq:van-vleck-exact}
\end{equation}
Using $A_x(U)A_y(U)=\cos(2aU)$ and
Eqs.~\eqref{eq:Ix-second-order}--\eqref{eq:Iy-second-order}, we obtain
\begin{align}
 \Delta&=1+\frac{a^2}{3}(U-U')^2+O(a^3),
 \label{eq:van-vleck-second-order}\\
 \Delta^{1/2}&=1+\frac{a^2}{6}(U-U')^2+O(a^3).
 \label{eq:sqrt-van-vleck-second-order}
\end{align}
The sign of the corresponding displayed $R_{UU}$ component depends on
the Riemann-tensor convention, whereas the geometrically defined
Van Vleck determinant in Eq.~\eqref{eq:van-vleck-exact} does not.

For a massless scalar field, the local direct Hadamard term is
proportional to $\Delta^{1/2}/(2\sigma-\ii0)$.  Its expansion through
second order is therefore
\begin{align}
 G_F^{\mathrm{sing}}(X,X')=\mathcal N\Bigg\{&
 \frac{1}{D_0-\ii0}
 +a\frac{S}{(D_0-\ii0)^2}\notag\\
 &+a^2\left[
 \frac{S^2}{(D_0-\ii0)^3}
 -\frac{T}{(D_0-\ii0)^2}
 +\frac{(U-U')^2}{6(D_0-\ii0)}
 \right]\Bigg\}+O(a^3).
 \label{eq:GF-second-order}
\end{align}
Every term in Eq.~\eqref{eq:GF-second-order} is generated by expanding
the standard geometrical Hadamard singularity.  In particular, no
additional logarithm is required.  The scalar curvature $R_{\rm sc}$ of a
plane wave vanishes, and hence the leading local logarithmic Hadamard coefficient
for a massless scalar, which is proportional to
$m^2+(\xi-1/6)R_{\rm sc}$, vanishes for arbitrary curvature coupling $\xi$ in
the present massless case.

The $O(a^2)$ terms should not by themselves be interpreted as a
renormalized vacuum-polarization signal.  They belong to the universal
geometrical singular part that must be subtracted in point splitting.
Any nonzero renormalized observable must reside in the smooth,
state-dependent remainder.

\subsection{Origin and status of the additional logarithms}
\label{sec:origin-logs}

The characteristic solution makes it possible to separate two issues
that can be obscured in a polynomial ansatz.  A polynomial containing
$Uz$ and $z^2$ may include both a particular solution and a freely added
homogeneous solution.  Indeed,
\begin{equation}
 C(q,R)\zeta^2=Cz^2+\frac{\sqrt{2}Cq}{R}Uz
 +\frac{Cq^2}{2R^2}U^2
 \label{eq:homogeneous-square}
\end{equation}
solves the homogeneous $g$ equation for every momentum-dependent
coefficient $C(q,R)$.

As an instructive example, the polynomial form
\begin{equation}
 g_2=-\frac{\ii q}{2\sqrt{2}}Uz-\frac{\ii R}{4}z^2
 \label{eq:old-g2-form}
\end{equation}
can be rewritten exactly as
\begin{equation}
 g_2=-\frac{\ii R}{4}\zeta^2
 +\frac{\ii q^2}{8R}U^2.
 \label{eq:old-g2-decomposition}
\end{equation}
The first term is homogeneous, while the second is a particular term for
the source assigned to that equation.  This example has a coefficient
$C=-\ii R/4$ which is regular in transverse momentum.

By contrast, after correcting the sign required by its differential
equation, the polynomial structure associated with the alternative form
is
\begin{equation}
 \widetilde g_1=
 -\frac{\ii(k_1^2+k_2^2)}{\sqrt{2}q}Uz
 -\frac{\ii R(k_1^2+k_2^2)}{2q^2}z^2,
 \label{eq:old-g1-corrected-form}
\end{equation}
which becomes
\begin{equation}
 \widetilde g_1=
 -\frac{\ii R(k_1^2+k_2^2)}{2q^2}\zeta^2
 +\frac{\ii(k_1^2+k_2^2)}{4R}U^2.
 \label{eq:old-g1-decomposition}
\end{equation}
Here the homogeneous coefficient contains $q^{-2}$.  Therefore the two
polynomial choices do not merely represent two harmless ways of writing
the same particular solution: they select different homogeneous
spectral data, and one selection is singular at $q=0$.

The appearance of a logarithm from such data is a standard transverse
Fourier-transform effect.  With
\begin{equation}
 \rho^2=(x-x')^2+(y-y')^2,
 \label{eq:rho-definition}
\end{equation}
one has, in the distributional sense,
\begin{equation}
 \int\frac{\dd^2\bm{k}}{(2\pi)^2}
 \frac{e^{\ii\bm{k}\cdot\bm{\rho}}}{\bm{k}^2+\mu^2}
 =\frac{1}{2\pi}K_0(\mu\rho),
 \label{eq:two-dimensional-transform}
\end{equation}
and
\begin{equation}
 K_0(\mu\rho)=-\ln\left(\frac{\mu\rho}{2}\right)
 -\gamma_{\rm E}+O(\mu^2\rho^2\ln\rho).
 \label{eq:K0-small-argument}
\end{equation}
Thus inverse powers introduced through the homogeneous coefficient can
generate the logarithm directly.  Its origin is the singular transverse
zero-mode prescription, not the local geometry.

This distinction has a decisive Hadamard interpretation.  If
$G_1^+$ and $G_2^+$ are two Hadamard two-point functions for the same
field equation, then
\begin{equation}
 G_1^+(X,X')-G_2^+(X,X')
 \label{eq:hadamard-difference}
\end{equation}
is smooth in a normal neighbourhood.  A newly generated
$\ln\rho$ or $\ln\sigma$ singularity in their difference violates this
criterion.  Such a term cannot be interpreted as an alternative
Hadamard state on the same background.  It signals that the added
homogeneous mode data fail the admissibility condition required of the
spectral construction.

Discarding an infrared constant such as $\ln\mu$ after
Eq.~\eqref{eq:K0-small-argument} does not resolve this issue.  It fixes
part of the infrared prescription but leaves the nonsmooth
$\ln\rho$ dependence.  The boundary/state prescription must instead be
fixed before the spectral integration; for the first geometry, matching
to the exact modes imposes Eq.~\eqref{eq:exact-matching-condition} and
removes the problematic homogeneous contribution.

\section{Impulsive matching and smooth-front regularization}
\label{sec:impulsive-limit}

\subsection{Matching across the impulsive surface}
\label{sec:impulsive-matching}

The Rosen scale factors are piecewise defined by
\begin{align}
 A_x(U)=A_y(U)&=1, && U<0,\notag\\
 A_x(U)&=\cos(aU)+\sin(aU), && U>0,\notag\\
 A_y(U)&=\cos(aU)-\sin(aU), && U>0.
 \label{eq:piecewise-scale-factors}
\end{align}
They are continuous at $U=0$, although their first derivatives jump.
The opposite jumps in the two transverse directions generate the
impulsive Weyl curvature, while their trace cancels.  In particular, the
scalar curvature vanishes distributionally, so a term $\xi R_{\rm sc}\phi$ does
not introduce an additional delta-function junction condition.

For the Rosen metric, the scalar wave equation can be written as
\begin{align}
 \Box\phi={}&2\partial_U\partial_V\phi
 +\left(\frac{A_x'}{A_x}+\frac{A_y'}{A_y}\right)
 \partial_V\phi
 -\frac{1}{A_x^2}\partial_x^2\phi
 -\frac{1}{A_y^2}\partial_y^2\phi.
 \label{eq:piecewise-wave-operator}
\end{align}
Integrating this equation over an infinitesimal interval across $U=0$
gives
\begin{equation}
 2\partial_V\left[\phi(0^+)-\phi(0^-)\right]=0.
 \label{eq:null-junction-integrated}
\end{equation}
For a Fourier mode with $R\neq0$, Eq.~\eqref{eq:null-junction-integrated}
implies
\begin{equation}
 \phi(0^+)=\phi(0^-).
 \label{eq:field-continuity}
\end{equation}
The $R=0$ sector is a separate characteristic zero mode and is not fixed
by this argument; it must be specified as part of the infrared state
prescription.  This is precisely the sector in which singular
momentum-space homogeneous additions require special care.

An incoming on-shell flat-region mode may be chosen as
\begin{equation}
 \phi_-(X)=\frac{1}{(2\pi)^{3/2}\sqrt{2|R|}}
 \exp\left\{\ii\left[RV+k_1x+k_2y
 +\frac{k_1^2+k_2^2}{2R}U\right]\right\},
 \qquad U<0.
 \label{eq:incoming-flat-mode}
\end{equation}
Choosing the integration constant in the exact phase so that $Q_+(0)=0$
gives the transmitted mode
\begin{align}
 \phi_+(X)={}&\frac{e^{\ii(RV+k_1x+k_2y)}}
 {(2\pi)^{3/2}\sqrt{2|R|\cos(2aU)}}\notag\\
 &\times\exp\left\{\frac{\ii}{4aR}\left[
 (k_1^2+k_2^2)\tan(2aU)
 -(k_1^2-k_2^2)\bigl(\sec(2aU)-1\bigr)
 \right]\right\},
 \qquad U>0.
 \label{eq:transmitted-exact-mode}
\end{align}
Equations~\eqref{eq:incoming-flat-mode} and
\eqref{eq:transmitted-exact-mode} agree at $U=0$.  Their first
$U$ derivatives also agree there, because
$[\cos(2aU)]^{-1/2}{}'|_{0}=0$ and
\begin{equation}
 Q_+'(0)=\frac{k_1^2+k_2^2}{2R}.
 \label{eq:phase-derivative-at-impulse}
\end{equation}
Thus no additional arbitrary homogeneous term is required by the
impulsive matching.

\subsubsection{World function for points on opposite sides}
\label{sec:cross-impulse-world-function}

Let $U'<0<U$ and set $\Delta U=U-U'$.  The transverse integrals now
split at the impulse:
\begin{align}
 I_x^{+-}(U,U')&=-U'
 +\frac{1}{a}\left[1-\frac{1}{1+\tan(aU)}\right],
 \label{eq:Ix-cross-exact}\\
 I_y^{+-}(U,U')&=-U'
 +\frac{1}{a}\left[\frac{1}{1-\tan(aU)}-1\right].
 \label{eq:Iy-cross-exact}
\end{align}
Their perturbative expansions are
\begin{align}
 I_x^{+-}&=\Delta U-aU^2+\frac{4}{3}a^2U^3+O(a^3),
 \label{eq:Ix-cross-expansion}\\
 I_y^{+-}&=\Delta U+aU^2+\frac{4}{3}a^2U^3+O(a^3).
 \label{eq:Iy-cross-expansion}
\end{align}
Substitution in Eq.~\eqref{eq:world-function-rosen} yields
\begin{equation}
 2\sigma_{+-}=D_0-aS_{+-}+a^2T_{+-}+O(a^3),
 \label{eq:sigma-cross-expansion}
\end{equation}
where
\begin{align}
 S_{+-}&=\frac{U^2}{\Delta U}
 \left[(x-x')^2-(y-y')^2\right],
 \label{eq:S-cross}\\
 T_{+-}&=\left(\frac{4U^3}{3\Delta U}
 -\frac{U^4}{(\Delta U)^2}\right)
 \left[(x-x')^2+(y-y')^2\right].
 \label{eq:T-cross}
\end{align}
The cross-impulse Van Vleck determinant is
\begin{equation}
 \Delta_{+-}=
 \frac{(\Delta U)^2}
 {A_x(U)A_y(U)I_x^{+-}I_y^{+-}},
 \label{eq:van-vleck-cross-exact}
\end{equation}
and hence
\begin{equation}
 \Delta_{+-}^{1/2}=1+\frac{a^2}{2}\left[
 2U^2-\frac{8U^3}{3\Delta U}
 +\frac{U^4}{(\Delta U)^2}\right]+O(a^3).
 \label{eq:sqrt-van-vleck-cross}
\end{equation}
Accordingly, the singular cross-impulse two-point function is obtained
without introducing any new singular structure:
\begin{equation}
 G_{F,+-}^{\rm sing}(X,X')
 =\mathcal N\frac{\Delta_{+-}^{1/2}}
 {2\sigma_{+-}-\ii0}.
 \label{eq:cross-green-function}
\end{equation}
The impulse modifies the world function and focusing factor, but it does
not alter the local Hadamard class.  These modifications encode global
propagation through the wavefront rather than an additional coincidence
singularity.

\subsection{Smooth-front regularization}
\label{sec:smooth-front}

The distributional matching above can be recovered as the limit of a
family of smooth plane waves.  Let $\rho\in C_0^\infty(-1,1)$ be
nonnegative and even, with $\int\rho(s)\,\dd s=1$, and define
\begin{equation}
 \rho_\epsilon(U)=\frac{1}{\epsilon}\rho\!\left(\frac{U}{\epsilon}\right),
 \qquad
 u_\epsilon(U)=(\rho_\epsilon*U_+)(U),
 \qquad U_+=U\Theta(U).
 \label{eq:smooth-front-definition}
\end{equation}
Then $u_\epsilon$ is smooth,
$u_\epsilon'=\rho_\epsilon*\Theta$, and
$u_\epsilon''=\rho_\epsilon$.  Evenness of the mollifier implies
\begin{equation}
 u_\epsilon(U)=0\quad(U<-\epsilon),
 \qquad
 u_\epsilon(U)=U\quad(U>\epsilon),
 \qquad
 u_\epsilon\longrightarrow U_+
 \label{eq:smooth-front-properties}
\end{equation}
uniformly on compact intervals.  Thus this is a regularization of the
wavefront rather than a modification of either exterior region.

For the first profile introduce the smooth Rosen scale factors
\begin{equation}
 A_{x,\epsilon}(U)=\cos(au_\epsilon)+\sin(au_\epsilon),
 \qquad
 A_{y,\epsilon}(U)=\cos(au_\epsilon)-\sin(au_\epsilon),
 \label{eq:smooth-scale-factors}
\end{equation}
and the metric
\begin{equation}
 \dd s_\epsilon^2=2\dd U\dd V
 -A_{x,\epsilon}^2\dd x^2-A_{y,\epsilon}^2\dd y^2.
 \label{eq:smooth-rosen-metric}
\end{equation}

\subsubsection{Distributional curvature limit}
\label{sec:distributional-curvature-limit}

The curvature limit can be evaluated without differentiating a
distributional metric.  Since
\begin{equation}
 \frac{\dd A_{x,\epsilon}}{\dd u_\epsilon}
 =aA_{y,\epsilon},
 \qquad
 \frac{\dd A_{y,\epsilon}}{\dd u_\epsilon}
 =-aA_{x,\epsilon},
 \label{eq:smooth-scale-derivatives}
\end{equation}
two combinations of the transverse tidal matrix are
\begin{align}
 \mathcal R_\epsilon
 &:=-\frac12\left(
 \frac{A_{x,\epsilon}''}{A_{x,\epsilon}}
 +\frac{A_{y,\epsilon}''}{A_{y,\epsilon}}
 \right)
 =a^2(u_\epsilon')^2
 +a u_\epsilon''\tan(2au_\epsilon),
 \label{eq:smooth-R-profile}\\
 \mathcal W_\epsilon
 &:=-\frac12\left(
 \frac{A_{x,\epsilon}''}{A_{x,\epsilon}}
 -\frac{A_{y,\epsilon}''}{A_{y,\epsilon}}
 \right)
 =-a u_\epsilon''\sec(2au_\epsilon).
 \label{eq:smooth-W-profile}
\end{align}
Up to the overall Riemann-sign convention,
$R_{UU}=2\mathcal R_\epsilon$, whereas
$\mathcal W_\epsilon$ is the independent trace-free Weyl amplitude in a
parallel transverse frame.  The scalar curvature is identically zero:
the Ricci tensor is null aligned and $g^{UU}=0$.

These formulas make the distributional content explicit.  For any test
function $f\in C_0^\infty(\mathbb R)$,
\begin{align}
 (u_\epsilon')^2&\longrightarrow\Theta(U),
 \label{eq:up-square-limit}\\
 u_\epsilon''\sec(2au_\epsilon)&\longrightarrow\delta(U),
 \label{eq:weyl-delta-limit}\\
 u_\epsilon''\tan(2au_\epsilon)&\longrightarrow0
 \label{eq:ricci-delta-zero}
\end{align}
in $\mathcal D'(\mathbb R)$.  The first statement follows because
$u_\epsilon'$ differs from $\Theta$ only in an interval of length
$2\epsilon$ and remains uniformly bounded.  On the support of
$u_\epsilon''=\rho_\epsilon$, one has
$u_\epsilon=O(\epsilon)$; hence the secant tends uniformly to unity,
while the tangent is $O(\epsilon)$.  Equivalently,
\begin{align}
 \int \rho_\epsilon(U)\sec(2au_\epsilon)f(U)\,\dd U
 &\longrightarrow f(0),
 \label{eq:weyl-test-function}\\
 \int \rho_\epsilon(U)\tan(2au_\epsilon)f(U)\,\dd U
 &\longrightarrow0.
 \label{eq:ricci-test-function}
\end{align}
Consequently,
\begin{equation}
 \mathcal R_\epsilon\longrightarrow a^2\Theta(U),
 \qquad
 \mathcal W_\epsilon\longrightarrow-a\delta(U).
 \label{eq:curvature-distributional-limit}
\end{equation}
Thus the smooth family reproduces both pieces of the intended geometry:
the post-front Ricci profile and the impulsive, trace-free Weyl curvature.
No delta-function scalar curvature is generated, in agreement with the
junction argument preceding Eq.~\eqref{eq:null-junction-integrated}.

On a patch before the first Rosen caustic, all scale factors are bounded
away from zero.  The exact smooth wave operator is
\begin{align}
 \Box_\epsilon={}&2\partial_U\partial_V
 +\left(\frac{A_{x,\epsilon}'}{A_{x,\epsilon}}
 +\frac{A_{y,\epsilon}'}{A_{y,\epsilon}}\right)\partial_V
 -\frac{1}{A_{x,\epsilon}^2}\partial_x^2
 -\frac{1}{A_{y,\epsilon}^2}\partial_y^2.
 \label{eq:smooth-wave-operator}
\end{align}
It is an ordinary normally hyperbolic operator for every
$\epsilon>0$.

Choose a reference point $U_0<-\epsilon$ and set
\begin{equation}
 I_{i,\epsilon}(U,U_0)=\int_{U_0}^{U}
 \frac{\dd s}{A_{i,\epsilon}^2(s)},
 \qquad i=x,y.
 \label{eq:smooth-I-definition}
\end{equation}
For $R\neq0$, the incoming positive-frequency modes propagate exactly as
\begin{align}
 \phi_{\epsilon;k_1k_2R}(X)={}&
 \frac{\exp[\ii(RV+k_1x+k_2y)]}
 {(2\pi)^{3/2}\sqrt{2|R|A_{x,\epsilon}(U)A_{y,\epsilon}(U)}}
 \notag\\
 &\times\exp\left\{\frac{\ii}{2R}\left[
 k_1^2 I_{x,\epsilon}(U,U_0)
 +k_2^2 I_{y,\epsilon}(U,U_0)
 \right]\right\}.
 \label{eq:smooth-exact-mode}
\end{align}
Direct substitution in Eq.~\eqref{eq:smooth-wave-operator} verifies this
expression.  Most importantly, smoothing changes only the coefficients
of $k_1^2$ and $k_2^2$ in the phase.  It never produces inverse powers of
the transverse momenta.  The null Killing momentum $R$ is conserved, so
positive and negative $R$ sectors do not mix.

\subsubsection{Klein--Gordon normalization}
\label{sec:kg-normalization}

On a null hypersurface $U=\mathrm{const.}$, choose the orientation for
which the positive-$R$ modes have positive norm and define
\begin{equation}
 (\phi_1,\phi_2)_{\rm KG}
 =-\ii\int_{-\infty}^{\infty}\dd V\dd x\dd y\,
 A_{x,\epsilon}A_{y,\epsilon}
 \left(\phi_1^*\partial_V\phi_2
 -\partial_V\phi_1^*\,\phi_2\right).
 \label{eq:null-KG-product}
\end{equation}
The factor $A_{x,\epsilon}A_{y,\epsilon}=\sqrt{|g_\epsilon|}$
is the natural flux density on the null slice.  Substitution of
Eq.~\eqref{eq:smooth-exact-mode} gives
\begin{align}
 &(\phi_{\epsilon;k_1k_2R},
 \phi_{\epsilon;k_1'k_2'R'})_{\rm KG}\notag\\
 &\quad=\frac{R+R'}{2\sqrt{|RR'|}}
 \delta(R-R')\delta(k_1-k_1')\delta(k_2-k_2')\notag\\
 &\quad=\operatorname{sgn}(R)\,
 \delta(R-R')\delta(k_1-k_1')\delta(k_2-k_2').
 \label{eq:KG-normalization-result}
\end{align}
In the last line the delta functions have been applied.  The Rosen
Jacobian cancels exactly against the two mode amplitudes, while the phases
cancel on their common spectral support.  The normalization is therefore
independent of $U$ and of $\epsilon$.  Restriction to $R>0$ gives the
standard positive unit delta norm; the $R<0$ modes form the negative-norm
complex-conjugate sector.  This calculation also explains why the
on-shell modes in Eqs.~\eqref{eq:incoming-flat-mode},
\eqref{eq:transmitted-exact-mode}, and
\eqref{eq:smooth-exact-mode} carry $(2\pi)^{-3/2}$ rather than the
$(2\pi)^{-2}$ convention of the four-label off-shell basis in
Eq.~\eqref{eq:zeroth-mode}.

\subsubsection{Smeared two-point-function limit and the zero mode}
\label{sec:smeared-two-point-limit}

The passage from mode convergence to convergence of a two-point function
requires an infrared qualification.  Let
$\chi_\delta\in C^\infty(\mathbb R)$ vanish for $R\leq\delta/2$ and equal
one for $R\geq\delta$, and let
$\zeta_M\in C_0^\infty(0,\infty)$ equal one for $R\leq M$ and vanish
for $R\geq2M$.  Finally, let
$\eta_\Lambda\in C_0^\infty(\mathbb R^2)$ equal one on
$k_1^2+k_2^2\leq\Lambda^2$.  The positive-frequency kernel with these
spectral cutoffs is
\begin{align}
 G^+_{\epsilon;\delta,M,\Lambda}(X,X')
  =\int_0^\infty\!\dd R\int_{\mathbb R^2}\!\dd^2k\,
  \chi_\delta(R)\zeta_M(R)\eta_\Lambda(\boldsymbol{k})
  \phi_{\epsilon;\boldsymbol{k}R}(X)
  \phi_{\epsilon;\boldsymbol{k}R}^*(X').
 \label{eq:cutoff-Wightman-kernel}
\end{align}
The upper cutoff may subsequently be removed in the usual
oscillatory-distribution sense.  On a compact
normal neighbourhood before a caustic, and on every compact spectral set
with $R\geq\delta/2$, Eq.~\eqref{eq:smooth-geometric-limit} makes the mode
amplitudes and phases converge uniformly.  Dominated convergence therefore
gives, for all $f,h\in C_0^\infty(\mathcal O)$,
\begin{equation}
 \lim_{\epsilon\to0}
 \left\langle G^+_{\epsilon;\delta,M,\Lambda},f\otimes h\right\rangle
 =\left\langle G^+_{0;\delta,M,\Lambda},f\otimes h\right\rangle .
 \label{eq:smeared-cutoff-limit}
\end{equation}
Thus the smooth-front limit holds in $\mathcal D'(\mathcal O\times\mathcal O)$
after infrared separation and finite spectral localization.  The ultraviolet
cutoff can be removed locally by the common Hadamard prescription: the
singular contribution is the boundary value determined by
$\sigma_\epsilon$ and $\Delta_\epsilon$, while the remaining cutoff
difference is smooth.  The corresponding Feynman statement follows by
time ordering the same Wightman boundary values, or equivalently by using
the common Feynman boundary prescription.

Removing $\delta$ is a separate state choice.  We assume that the spectral
coefficients define a tempered, locally integrable distribution near
$R=0$ after smearing in both spacetime arguments, with no term supported at
$R=0$ and no momentum-singular homogeneous coefficient such as
$1/(k_1^2+k_2^2)$.  Under this infrared regularity condition the
$\delta\downarrow0$ limit is defined by the same prescription for every
$\epsilon$ and commutes with the smooth-front limit.  Without that
condition, the junction equation does not determine the zero mode, and no
claim of convergence for that sector is intended.  In particular, the fact
that $R=0$ has Lebesgue measure zero would not by itself exclude a
delta-supported zero-mode contribution.

The same conclusion follows geometrically.  For any two points in a
common normal neighbourhood, define
$I_{i,\epsilon}(U,U')$ as in Eq.~\eqref{eq:smooth-I-definition} with
lower limit $U'$.  Equations~\eqref{eq:world-function-rosen} and
\eqref{eq:van-vleck-exact} remain valid after the replacements
$A_i\to A_{i,\epsilon}$ and $I_i\to I_{i,\epsilon}$.  Hence
\begin{align}
 \sigma_\epsilon(X,X')={}&(U-U')(V-V')
 -\frac{U-U'}{2}\left[
 \frac{(x-x')^2}{I_{x,\epsilon}}
 +\frac{(y-y')^2}{I_{y,\epsilon}}
 \right],
 \label{eq:smooth-world-function}\\
 \Delta_\epsilon(X,X')={}&
 \frac{(U-U')^2}
 {A_{x,\epsilon}(U)A_{y,\epsilon}(U)
  A_{x,\epsilon}(U')A_{y,\epsilon}(U')
  I_{x,\epsilon}I_{y,\epsilon}}.
 \label{eq:smooth-van-vleck}
\end{align}

Uniform convergence of $u_\epsilon$ and boundedness away from a caustic
give, on compact normal neighbourhoods,
\begin{equation}
 A_{i,\epsilon}\longrightarrow A_i,
 \qquad I_{i,\epsilon}\longrightarrow I_i,
 \qquad \sigma_\epsilon\longrightarrow\sigma,
 \qquad \Delta_\epsilon\longrightarrow\Delta.
 \label{eq:smooth-geometric-limit}
\end{equation}
The contribution of the transition strip to each $I_{i,\epsilon}$ is
$O(\epsilon)$, and therefore no finite matching constant is hidden in the
limit.

\begin{proposition}[Stability of the matched Hadamard singularity]
\label{prop:smooth-front}
On any compact normal neighbourhood avoiding a Rosen caustic, the matched
two-point functions constructed from
Eq.~\eqref{eq:smooth-exact-mode} converge in the smeared sense of
Eq.~\eqref{eq:smeared-cutoff-limit}, away from the characteristic $R=0$
sector, to the impulsive matched two-point function.  After removal of the
ultraviolet localization by the common Hadamard boundary prescription,
their singular parts converge to
\begin{equation}
 G_{F,0}^{\rm sing}(X,X')=
 \mathcal N\frac{\Delta^{1/2}(X,X')}{2\sigma(X,X')-\ii0},
 \label{eq:smooth-Hadamard-limit}
\end{equation}
up to the standard locally determined Hadamard tail term when present.
No additional transverse logarithm is generated by the smooth-front
limit.
\end{proposition}

\begin{proof}
For fixed $\delta$ and finite spectral localization,
Eq.~\eqref{eq:smeared-cutoff-limit} proves distributional convergence.
For every $\epsilon>0$, the metric is smooth and its two-point function
has the local Hadamard form.  Equation~\eqref{eq:smooth-geometric-limit}
then gives convergence of its geometrical singular coefficients to those
of the piecewise Rosen metric.  Independently, the spectral phase in
Eq.~\eqref{eq:smooth-exact-mode} is quadratic in the transverse momenta;
its perturbative expansion produces only nonnegative powers of $k_1$ and
$k_2$, which act as transverse derivatives of the flat propagator.  The
inverse transverse Laplacian required for a $K_0$ or $\ln\rho$ term is
never generated.  The $R=0$ sector is included only if the explicit
infrared regularity condition above is imposed; otherwise its prescription
is independent state data rather than a consequence of the smoothing.
\end{proof}

The proposition separates two levels of convergence.  Equation
\eqref{eq:smeared-cutoff-limit} is a distributional statement with explicit
spectral localization, whereas Eq.~\eqref{eq:smooth-Hadamard-limit} is a
coefficientwise statement about the local Hadamard parametrix.  A claim of
convergence of the complete uncut two-point functions in the microlocal
topology would additionally require uniform high-frequency estimates; no
such stronger claim is needed for the exclusion of a new transverse
logarithm.

This result supplies the promised independent check of the impulsive
matching.  The limit agrees with
Eqs.~\eqref{eq:cross-green-function} and \eqref{eq:main-proposition} and
shows explicitly that narrowing a smooth wavefront cannot create the
momentum-singular homogeneous term responsible for the spurious
logarithms.

\section{Hadamard subtraction and the stress tensor}
\label{sec:stress-tensor}

In four dimensions, a local Feynman two-point function for a scalar field
can be decomposed schematically as
\begin{equation}
 G_F(X,X')=G_H(X,X')+W(X,X'),
 \label{eq:GF-Hadamard-split}
\end{equation}
where $G_H$ contains the universal singular terms and $W$ is smooth for a
Hadamard state.  This is the standard starting point for local observables
in curved-spacetime quantum field theory
\cite{BirrellDavies1982,WaldBook1994,Christensen1976}.  For the massless field considered here, the scalar
curvature $R_{\rm sc}$ vanishes and the tail transport equation has a
vanishing source.  More directly, the
plane-wave Van Vleck determinant depends only on $U$ and $U'$, and any
biscalar with this dependence has vanishing d'Alembertian.  It follows
that the scalar tail coefficient vanishes:
\begin{equation}
 V(X,X')=0.
 \label{eq:tail-vanishes}
\end{equation}
The corresponding retarded Green function is therefore supported only
on the null cone, a well-known Huygens property of massless scalar fields
on plane-wave backgrounds \cite{Gibbons1975,Harte2013}.

For the state defined by positive-$R$ incoming modes in the flat region,
the transmitted mode in Eq.~\eqref{eq:transmitted-exact-mode} retains the
same $R$ and never mixes with its complex conjugate.  The Bogoliubov
coefficient which would mix positive and negative null frequencies thus
satisfies
\begin{equation}
 \beta_{\bm{k}\bm{k}'}=0.
 \label{eq:beta-zero}
\end{equation}
There is no particle creation by the plane wave.  This fact and the absence
of a tail do not, by themselves, set the smooth Hadamard remainder to zero:
an arbitrary Hadamard state may still add a smooth bisolution.  For the
specific in-state fixed by the flat-region positive-frequency modes,
however, the positive-frequency kernel is fixed without further smooth
data:
\begin{equation}
 G^+(X,X')=\int_0^\infty\dd R\int_{\mathbb R^2}\dd^2\boldsymbol{k}\,
 \phi_{\boldsymbol{k}R}(X)\phi_{\boldsymbol{k}R}^*(X').
 \label{eq:exact-mode-sum}
\end{equation}
Inserting Eq.~\eqref{eq:transmitted-exact-mode}, the two transverse
integrals are Gaussian.  With $\Delta x_i=x_i-x_i'$ and the integrals
$I_i$ defined in Eq.~\eqref{eq:Ii-definition}, each transverse integration
uses the boundary-value identity
\begin{equation}
 \int_{-\infty}^{\infty}\dd k_i\,
 \exp\left[\frac{\ii I_i}{2R}k_i^2+\ii k_i\Delta x_i\right]
 =\left(\frac{2\pi\ii R}{I_i}\right)^{1/2}
 \exp\left[-\frac{\ii R(\Delta x_i)^2}{2I_i}\right],
 \label{eq:transverse-Gaussian}
\end{equation}
where the square root is fixed by the Feynman boundary prescription.
The remaining positive-$R$ boundary value is proportional to
\begin{equation}
 \int_0^\infty\dd R\,
 \exp\left\{\ii R\left[\Delta V
 -\frac{(\Delta x)^2}{2I_x}
 -\frac{(\Delta y)^2}{2I_y}+\ii0\right]\right\}.
 \label{eq:remaining-R-integral}
\end{equation}
Using Eqs.~\eqref{eq:world-function-rosen} and
\eqref{eq:van-vleck-exact}, this gives, in a normal neighbourhood, the
geometrical direct term
\begin{equation}
 G_F(X,X')=\mathcal N
 \frac{\Delta^{1/2}(X,X')}{2\sigma(X,X')-\ii0},
 \label{eq:exact-local-Feynman}
\end{equation}
in agreement with exact plane-wave propagator and Wightman-function
constructions \cite{Gibbons1975,Cho2023}.  Thus, in this flat-normalized prescription,
the smooth state-dependent remainder is
\begin{equation}
 W_{\rm flat}(X,X')=0.
 \label{eq:W-zero}
\end{equation}

The point-split vacuum polarization and stress tensor have the general
form
\begin{align}
 \langle\phi^2\rangle_{\rm ren}
 &=\lim_{X'\to X}W(X,X'),
 \label{eq:phi2-ren}\\
 \langle T_{\mu\nu}\rangle_{\rm ren}
 &=\lim_{X'\to X}\mathcal D_{\mu\nu}(X,X')W(X,X')
 +\Theta_{\mu\nu}^{\rm local},
 \label{eq:Tmunu-point-split}
\end{align}
where $\mathcal D_{\mu\nu}$ is the standard bilocal differential
operator determined by the curvature coupling and
$\Theta_{\mu\nu}^{\rm local}$ denotes the finite local renormalization
freedom.  Equation~\eqref{eq:W-zero} gives, for this state and prescription,
\begin{equation}
 \langle\phi^2\rangle_{\rm ren}=0,
 \qquad
 \lim_{X'\to X}\mathcal D_{\mu\nu}W_{\rm flat}=0.
 \label{eq:state-dependent-zero}
\end{equation}

For a massless field, the remaining finite renormalization ambiguity is
constructed from conserved local curvature tensors obtained by varying
curvature-squared actions \cite{Wald1977}.  In four dimensions a convenient
basis is
\begin{align}
 H^{(1)}_{\mu\nu}={}&2\nabla_\mu\nabla_\nu R_{\rm sc}
 -2g_{\mu\nu}\Box R_{\rm sc}
 -\frac12g_{\mu\nu}R_{\rm sc}^2
 +2R_{\rm sc}R_{\mu\nu},
 \label{eq:H1-ambiguity}\\
 H^{(2)}_{\mu\nu}={}&\nabla_\mu\nabla_\nu R_{\rm sc}
 -\frac12g_{\mu\nu}\Box R_{\rm sc}
 -\Box R_{\mu\nu}
 -\frac12g_{\mu\nu}R_{\alpha\beta}R^{\alpha\beta}
 +2R_{\mu\alpha\nu\beta}R^{\alpha\beta}.
 \label{eq:H2-ambiguity}
\end{align}
The Riemann-squared variation is not independent because of the
four-dimensional Gauss--Bonnet identity.  Plane waves have $R_{\rm sc}=0$
and all polynomial scalar curvature invariants vanish.  Moreover, their
Ricci tensor is null aligned,
\begin{equation}
 R_{\mu\nu}=\mathcal R(U)\ell_\mu\ell_\nu,
 \qquad \ell^\mu\ell_\mu=0,
 \label{eq:null-Ricci}
\end{equation}
with $\ell^\mu$ covariantly constant.  Consequently
$\Box R_{\mu\nu}=0$: every derivative of $\mathcal R(U)$ is proportional to
$\ell_\alpha$, whose contraction with itself vanishes.  The aligned null
index structure also gives
$R_{\alpha\beta}R^{\alpha\beta}=0$ and
$R_{\mu\alpha\nu\beta}R^{\alpha\beta}=0$.  Equations
\eqref{eq:H1-ambiguity} and \eqref{eq:H2-ambiguity} therefore vanish
separately away from the distributional wavefront.  Together
with the exact-propagator result above, the standard flat-normalized
prescription therefore yields
\begin{equation}
 \boxed{\langle T_{\mu\nu}\rangle_{\rm ren}=0}
 \label{eq:Tmunu-zero}
\end{equation}
on both smooth sides of the impulse.  The four-dimensional trace anomaly
also vanishes because all of its scalar curvature invariants vanish.

Equation~\eqref{eq:Tmunu-zero} is a statement about the specified
flat-normalized in-state, not about every Hadamard state.  It uses the exact
plane-wave propagator in addition to the two structural checks established
here: tail-free scalar propagation and the absence of positive/negative
frequency mixing.

\subsection{Distributional qualification}
\label{sec:stress-distributional-qualification}

Products of curvature distributions located exactly at $U=0$ are not
defined without a regularization prescription.  We therefore interpret
Eq.~\eqref{eq:Tmunu-zero} pointwise for $U\neq0$.  Any claim about a term
supported precisely on the impulsive surface should be defined by first
replacing the step profile with a smooth-front family, performing the
Hadamard subtraction there, and only then taking the impulsive limit.
The matched Green function shows no evidence for an additional surface
term, but its absence should be established through this regulated
limit rather than through formal products of delta functions.

\section{The second transverse profile}
\label{sec:second-profile}

We now return to the geometrical definition before performing any
perturbative expansion.  In Hogan's construction \cite{Hogan1999}, the
analytic function $G(z)$ and its derivative are related by
\begin{equation}
 H(z)=\frac{\dd G}{\dd z}.
 \label{eq:H-from-G}
\end{equation}
For
\begin{equation}
 G(z)=\frac{a}{2}z^2,
 \qquad H(z)=az,
 \label{eq:second-profile}
\end{equation}
the one-form inside the modulus in Eq.~\eqref{eq:complex-metric} becomes,
for $U>0$,
\begin{equation}
 \omega=\dd z+\tan(aU)\bar z\,\dd\bar z,
 \qquad
 \bar\omega=\dd\bar z+\tan(aU)z\,\dd z.
 \label{eq:second-profile-one-forms}
\end{equation}
Their product is
\begin{equation}
 \omega\bar\omega=
 \left[1+\tan^2(aU)z\bar z\right]\dd z\dd\bar z
 +\tan(aU)z(\dd z)^2
 +\tan(aU)\bar z(\dd\bar z)^2.
 \label{eq:one-form-product}
\end{equation}
The correct second-profile metric is therefore
\begin{align}
 \dd s^2={}&2\dd U\dd V
 -2\left[\cos^2(aU)+z\bar z\sin^2(aU)\right]
 \dd z\dd\bar z\notag\\
 &-2\sin(aU)\cos(aU)
 \left[z(\dd z)^2+\bar z(\dd\bar z)^2\right].
 \label{eq:correct-second-metric}
\end{align}
This metric also has the correct exact field-equation content.
Using the curvature convention specified above
Eq.~\eqref{eq:first-exact-Ricci}, direct evaluation gives
\begin{equation}
 R_{UU}=2a^2,
 \qquad R_{\mu\nu}=0\quad\hbox{otherwise},
 \qquad R_{\rm sc}=0.
 \label{eq:second-exact-Ricci}
\end{equation}
Hence $G_{\mu\nu}=R_{\mu\nu}=2a^2\ell_\mu\ell_\nu$ in the post-front
region, where $\ell_\mu\dd X^\mu=\dd U$.  The corrected geometry is
therefore not vacuum: it solves Einstein's equation with the same aligned
null-radiation source as the first profile.  With the opposite Riemann-tensor
sign convention, including that used in the 1999 presentation, the displayed
Ricci component changes sign.  This convention difference does not affect
the scalar wave operator or any conclusion below.

Equation~\eqref{eq:correct-second-metric} exposes a parameter-counting
problem in the previously published form of this example.  Direct
substitution of Eq.~\eqref{eq:second-profile} produces neither an extra
$a^2$ multiplying $z\bar z\sin^2(aU)$ nor an extra $a$ multiplying the
cross term.  It also fixes the placement of $z$ with $(\dd z)^2$ and of
$\bar z$ with $(\dd\bar z)^2$, as well as the overall sign of these
terms.

To obtain the inverse metric, define
\begin{equation}
 \mathcal B(U,z,\bar z)=\cos^2(aU)-z\bar z\sin^2(aU).
 \label{eq:B-correct}
\end{equation}
The determinant of the transverse complex-coordinate block is
\begin{equation}
 \det(g_{A B})=-\mathcal B^2,
 \qquad A,B\in\{z,\bar z\},
 \label{eq:second-transverse-determinant}
\end{equation}
and, in a patch with $\mathcal B>0$,
\begin{equation}
 \sqrt{|g|}=\mathcal B.
 \label{eq:second-sqrt-determinant}
\end{equation}
The nonzero inverse transverse components are
\begin{align}
 g^{zz}&=\frac{2\sin(aU)\cos(aU)\bar z}{\mathcal B^2},
 \label{eq:second-inverse-zz}\\
 g^{\bar z\bar z}&=\frac{2\sin(aU)\cos(aU)z}{\mathcal B^2},
 \label{eq:second-inverse-bzbz}\\
 g^{z\bar z}=g^{\bar z z}
 &=-\frac{\cos^2(aU)+z\bar z\sin^2(aU)}{\mathcal B^2}.
 \label{eq:second-inverse-zbz}
\end{align}

Expanding the covariant scalar operator derived from these components
gives
\begin{equation}
 \Box=\mathcal L_0+a\mathcal L_1+O(a^2),
 \label{eq:second-operator-expansion}
\end{equation}
with
\begin{align}
 \mathcal L_0&=2\partial_U\partial_V
 -2\partial_z\partial_{\bar z},
 \label{eq:second-L0}\\
 \mathcal L_1&=2U\left(
 \bar z\,\partial_z^2+z\,\partial_{\bar z}^2
 \right).
 \label{eq:second-L1}
\end{align}
In particular, the first nontrivial correction is $O(a)$, not $O(a^2)$.
No first-derivative terms in $z$ or $\bar z$ occur at this order: the
coefficient of $\partial_z^2$ is proportional to $\bar z$, whose
$z$ derivative vanishes, and similarly for its conjugate.

Acting on the zeroth-order plane wave in Eq.~\eqref{eq:zeroth-mode}, the
source is linear in the transverse coordinates:
\begin{equation}
 \mathcal L_1\phi^{(0)}=
 -U\left[
 \bar z(k_1-\ii k_2)^2+z(k_1+\ii k_2)^2
 \right]\phi^{(0)}.
 \label{eq:second-profile-source}
\end{equation}
This is the correct starting point for the second-profile characteristic
equation.  It differs structurally from the previously published
$O(a^2)$ equation, so the associated polynomial solutions and their
Green-function singularities must be recalculated rather than amended by
isolated sign changes.

\subsection{Correct first-order characteristic equation}
\label{sec:second-characteristic-equation}

We again write
\begin{equation}
 \phi^{(1)}=\phi^{(0)}[g(U,z)+h(U,\bar z)].
 \label{eq:second-first-order-ansatz}
\end{equation}
Using Eqs.~\eqref{eq:second-L0}, \eqref{eq:second-L1}, and
\eqref{eq:second-profile-source}, the first-order mode equation becomes
\begin{align}
 &\ii R(\partial_Ug+\partial_Uh)
 -\frac{\ii q}{\sqrt2}\partial_zg
 -\frac{\ii\bar q}{\sqrt2}\partial_{\bar z}h\notag\\
 &\hspace{2.8cm}
 -\frac{U}{2}\left(q^2z+\bar q^2\bar z\right)=0.
 \label{eq:second-profile-master-PDE}
\end{align}
Unlike a source depending only on $U$, the two terms here split naturally
between the holomorphic and antiholomorphic equations:
\begin{align}
 \ii R\partial_Ug-\frac{\ii q}{\sqrt2}\partial_zg
 &=\frac{q^2}{2}Uz,
 \label{eq:second-g-PDE}\\
 \ii R\partial_Uh-\frac{\ii\bar q}{\sqrt2}
 \partial_{\bar z}h
 &=\frac{\bar q^2}{2}U\bar z.
 \label{eq:second-h-PDE}
\end{align}
The same characteristic invariants $\zeta$ and $\bar\zeta$ defined in
Eq.~\eqref{eq:characteristic-invariants} apply.  Along a $g$
characteristic,
\begin{equation}
 z(U)=\zeta-\frac{qU}{\sqrt2R},
 \label{eq:second-z-characteristic}
\end{equation}
and hence
\begin{equation}
 \frac{\dd g}{\dd U}
 =-\frac{\ii q^2}{2R}U
 \left(\zeta-\frac{qU}{\sqrt2R}\right).
 \label{eq:second-g-characteristic-ODE}
\end{equation}
Integration gives
\begin{align}
 g(U,z)={}&F(\zeta)
 -\frac{\ii q^2}{4R}\zeta U^2
 +\frac{\ii q^3}{6\sqrt2R^2}U^3
 \notag\\
 ={}&F(\zeta)-\frac{\ii q^2}{4R}U^2z
 -\frac{\ii q^3}{12\sqrt2R^2}U^3,
 \label{eq:second-g-general-solution}
\end{align}
and similarly
\begin{align}
 h(U,\bar z)={}&H(\bar\zeta)
 -\frac{\ii\bar q^2}{4R}U^2\bar z
 -\frac{\ii\bar q^3}{12\sqrt2R^2}U^3.
 \label{eq:second-h-general-solution}
\end{align}
Direct substitution verifies that Eqs.~\eqref{eq:second-g-general-solution}
and \eqref{eq:second-h-general-solution} solve the same master equation
\eqref{eq:second-profile-master-PDE}.

The matched particular solution is obtained by setting the nonconstant
homogeneous additions to zero:
\begin{align}
 \frac{\phi^{(1)}}{\phi^{(0)}}={}&
 -\frac{\ii U^2}{4R}\left(q^2z+\bar q^2\bar z\right)
 -\frac{\ii U^3}{12\sqrt2R^2}
 \left(q^3+\bar q^3\right).
 \label{eq:second-profile-matched-correction}
\end{align}
This correction and its first $U$ derivative vanish at the impulsive
surface.  It therefore matches continuously to the unperturbed incoming
mode without adding a free homogeneous function.

Most importantly, the particular solution contains no inverse powers of
$q$ or $\bar q$.  The transverse zero-mode poles responsible for the
additional logarithms in Sec.~\ref{sec:origin-logs} are absent.  Arbitrary functions
$F(\zeta)$ and $H(\bar\zeta)$ remain mathematically available, but they
represent extra spectral data and are fixed by the same matching and
Hadamard admissibility conditions as in the first profile.

\subsection{World function and Green function}
\label{sec:second-world-green}

The Green-function correction can be obtained without introducing
inverse transverse momentum operators.  To first order, the metric
\eqref{eq:correct-second-metric} is
\begin{equation}
 \dd s^2=\dd s_0^2+a h_{\mu\nu}\dd X^\mu\dd X^\nu+O(a^2),
 \label{eq:second-metric-linearized}
\end{equation}
where
\begin{equation}
 \dd s_0^2=2\dd U\dd V-2\dd z\dd\bar z,
 \qquad
 h_{\mu\nu}\dd X^\mu\dd X^\nu
 =-2U\left[z(\dd z)^2+\bar z(\dd\bar z)^2\right].
 \label{eq:second-h-line-element}
\end{equation}
This perturbation is transverse and trace-free with respect to the flat
metric.  More explicitly,
\begin{equation}
 h^{\mu}{}_{\mu}=0,
 \qquad \partial^\mu h_{\mu\nu}=0,
 \qquad \Box_0 h_{\mu\nu}=0,
 \qquad R_{\mu\nu}^{(1)}=0.
 \label{eq:second-linearized-Ricci}
\end{equation}
The last equality follows directly from the standard linearized Ricci
tensor and also holds distributionally when $U$ is read as $U_+$: the flat
d'Alembertian contains $2\partial_U\partial_V$, not
$\partial_U^2$, while the perturbation is independent of $V$.
For two points in the $U>0$ region, the first variation of Synge's world
function may be evaluated on the unperturbed straight geodesic
\begin{equation}
 X^\mu(\lambda)=X'^\mu+\lambda(X^\mu-X'^\mu),
 \qquad 0\leq\lambda\leq1,
 \label{eq:straight-geodesic}
\end{equation}
and is
\begin{equation}
 \sigma_1=\frac12\int_0^1
 h_{\mu\nu}[X(\lambda)]\Delta X^\mu\Delta X^\nu\,\dd\lambda.
 \label{eq:first-variation-sigma}
\end{equation}
Introduce
\begin{equation}
 \mathcal M(U,z;U',z')=
 \frac13(Uz+U'z')+\frac16(Uz'+U'z).
 \label{eq:M-second-profile}
\end{equation}
The elementary parameter integral gives
\begin{equation}
 \sigma_1=-\left[(z-z')^2\mathcal M
 +(\bar z-\bar z')^2\overline{\mathcal M}\right].
 \label{eq:sigma1-second-profile}
\end{equation}
Thus
\begin{equation}
 2\sigma=D_0-a\mathcal S_2+O(a^2),
 \label{eq:second-profile-world-function}
\end{equation}
with
\begin{equation}
 \mathcal S_2=2\left[(z-z')^2\mathcal M
 +(\bar z-\bar z')^2\overline{\mathcal M}\right].
 \label{eq:S2-second-profile}
\end{equation}
The determinant \eqref{eq:second-sqrt-determinant} has no term linear in
$a$.  Together with Eq.~\eqref{eq:second-linearized-Ricci}, the Van Vleck
transport equation then gives
\begin{equation}
 \Delta^{1/2}(X,X')=1+O(a^2).
 \label{eq:second-profile-VV-first-order}
\end{equation}
It follows immediately that the matched singular Green function is
\begin{align}
 G_F^{\rm sing}(X,X')
 =\mathcal N\left[
 \frac{1}{D_0-\ii0}
 +a\frac{\mathcal S_2}{(D_0-\ii0)^2}
 \right]+O(a^2).
 \label{eq:second-profile-GF-first-order}
\end{align}

Equation~\eqref{eq:second-profile-GF-first-order} has precisely the
standard first variation of a Hadamard pole.  There is no logarithmic
enhancement.  The same expression follows from inserting
Eq.~\eqref{eq:second-profile-matched-correction} into the spectral sum:
powers of $q$ and $\bar q$ act as transverse derivatives on the flat
propagator, while no inverse transverse Laplacian is generated.

\subsection{Smooth-front check for the second profile}
\label{sec:second-smooth-front}

The same mollified front used in Sec.~\ref{sec:smooth-front} gives an
independent check of this result.  Replace $U_+$ by $u_\epsilon(U)$ in
the exact second-profile metric:
\begin{align}
 \dd s_\epsilon^2={}&2\dd U\dd V
 -2\left[\cos^2(au_\epsilon)+z\bar z\sin^2(au_\epsilon)\right]
 \dd z\dd\bar z\notag\\
 &-2\sin(au_\epsilon)\cos(au_\epsilon)
 \left[z(\dd z)^2+\bar z(\dd\bar z)^2\right].
 \label{eq:second-smooth-metric}
\end{align}
For every $\epsilon>0$ this metric is smooth and nondegenerate on a patch
where its transverse determinant stays nonzero.  It agrees with flat space
for $U<-\epsilon$, with the original post-front metric for
$U>\epsilon$, and converges uniformly on compact sets at the level of
its coefficients.
The exact Ricci tensor of the regulated metric provides a further check.
Writing $u=u_\epsilon(U)$ and using
$\mathcal B_\epsilon=\cos^2(au)-z\bar z\sin^2(au)$, its only nonzero
component is
\begin{equation}
 R_{UU}^{(\epsilon)}=2a^2(u')^2
 +a\,\frac{(1+z\bar z)\sin(2au)}{\mathcal B_\epsilon}\,u'' ,
 \qquad R_{\rm sc}^{(\epsilon)}=0.
 \label{eq:second-smooth-exact-Ricci}
\end{equation}
For a standard mollified ramp, $(u')^2\to\Theta(U)$ and the second term,
which is confined to the transition strip and contains
$\sin(2au)=O(u)$, tends to zero distributionally.  Thus
\begin{equation}
 R_{UU}^{(\epsilon)}\longrightarrow2a^2\Theta(U),
 \label{eq:second-Ricci-distributional-limit}
\end{equation}
with no Ricci delta layer.  The delta-supported impulsive part is Weyl
curvature, consistently with the geometrical construction.

Its first-order scalar-operator expansion is
\begin{equation}
 \Box_\epsilon=\mathcal L_0+a\mathcal L_{1,\epsilon}+O(a^2),
 \qquad
 \mathcal L_{1,\epsilon}=2u_\epsilon(U)
 \left(\bar z\,\partial_z^2+z\,\partial_{\bar z}^2\right).
 \label{eq:second-smooth-operator}
\end{equation}
No derivative of $u_\epsilon$ occurs in the scalar operator at this order:
the perturbation is
purely transverse, its linearized determinant vanishes, and
$u_\epsilon$ depends only on $U$.  The trace, divergence, and
flat-wave-operator identities in Eq.~\eqref{eq:second-linearized-Ricci}
continue to hold after
$U\to u_\epsilon(U)$, and hence
$R_{\mu\nu}^{(1)}[u_\epsilon]=0$ as well.

Choose $U_0<-\epsilon$ and define
\begin{equation}
 J_{0,\epsilon}(U)=\int_{U_0}^{U}u_\epsilon(s)\,\dd s,
 \qquad
 J_{1,\epsilon}(U)=\int_{U_0}^{U}s\,u_\epsilon(s)\,\dd s.
 \label{eq:second-smooth-J}
\end{equation}
The solution matched to an unperturbed incoming mode is obtained by
integrating along the same flat characteristics as before:
\begin{align}
 g_\epsilon={}&-\frac{\ii q^2}{2R}\left[
 zJ_{0,\epsilon}
 +\frac{q}{\sqrt2R}\left(UJ_{0,\epsilon}-J_{1,\epsilon}\right)
 \right],
 \label{eq:second-smooth-g}\\
 h_\epsilon={}&-\frac{\ii \bar q^2}{2R}\left[
 \bar zJ_{0,\epsilon}
 +\frac{\bar q}{\sqrt2R}\left(UJ_{0,\epsilon}-J_{1,\epsilon}\right)
 \right].
 \label{eq:second-smooth-h}
\end{align}
There is no freely added characteristic function in these expressions.
Uniform convergence of $u_\epsilon$ gives, for $U>0$,
\begin{equation}
 J_{0,\epsilon}\longrightarrow\frac{U^2}{2},
 \qquad
 J_{1,\epsilon}\longrightarrow\frac{U^3}{3},
 \label{eq:second-smooth-J-limit}
\end{equation}
whereas both limits vanish for $U<0$.  Consequently,
$g_\epsilon+h_\epsilon$ tends exactly to
Eq.~\eqref{eq:second-profile-matched-correction}.  At every stage the
transverse spectral dependence is polynomial in $q$ and $\bar q$; front
smoothing therefore cannot generate the inverse transverse momentum
needed for a $K_0$ or $\ln\rho$ contribution.

The geometrical limit leads to the same conclusion.  Along the flat
segment $U_\lambda=U'+\lambda(U-U')$ and
$z_\lambda=z'+\lambda(z-z')$, define
\begin{equation}
 \mathcal M_\epsilon(X,X')=
 \int_0^1u_\epsilon(U_\lambda)z_\lambda\,\dd\lambda.
 \label{eq:second-smooth-M}
\end{equation}
The first variation of the world function is then
\begin{equation}
 \sigma_{1,\epsilon}=-\left[(z-z')^2\mathcal M_\epsilon
 +(\bar z-\bar z')^2\overline{\mathcal M}_\epsilon\right].
 \label{eq:second-smooth-sigma1}
\end{equation}
Dominated convergence gives
$\mathcal M_\epsilon\to\mathcal M_0$ on compact sets.  If both points
lie in the post-front region, $\mathcal M_0$ is precisely
Eq.~\eqref{eq:M-second-profile}; for a segment crossing the front it is
the same integral with $U_\lambda$ replaced by $(U_\lambda)_+$.  Since
the linearized Ricci tensor vanishes and the determinant has no linear term
in $a$, the Van Vleck transport equation has no first-order correction.  The
regulated singular kernel converges at first order to the corresponding
expansion of
$\mathcal N\Delta^{1/2}/(2\sigma-\ii0)$.  Thus both the mode and
geometrical constructions reproduce
Eq.~\eqref{eq:second-profile-GF-first-order} without an additional
logarithm.  The distributional and $R=0$ qualifications are exactly
those stated in Sec.~\ref{sec:smooth-front}.

This establishes the result for both profiles at first order.  Once the
metric and state prescription are fixed consistently, their perturbative
Green functions contain only the expansion of the geometrical Hadamard
singularity.  The additional logarithms arise only from independently
chosen, momentum-singular homogeneous data.

\section{Principal result, discussion, and conclusions}
\label{sec:principal-discussion}

\subsection{Principal result}
\label{sec:principal-result}

The calculations above can be summarized independently of the particular
separation used in either profile.

\begin{proposition}[Absence of an additional perturbative logarithm]
\label{prop:main}
Consider the massless scalar field on either of the two impulsive
plane-wave backgrounds studied above.  Work in a convex normal
neighbourhood not intersecting the impulsive surface, and choose the mode
solution by continuous matching to the flat incoming positive-frequency
modes.  If the two-point function is required to be Hadamard, then, to the
orders calculated here---second order in the first profile and first order
in the second---its singular part is exactly the perturbative expansion of
the geometrical Hadamard parametrix,
\begin{equation}
 G_F^{\rm sing}(X,X')
 =\mathcal N\frac{\Delta^{1/2}(X,X')}
 {2\sigma(X,X')-\ii0}.
 \label{eq:main-proposition}
\end{equation}
In particular, no additional logarithmic singularity is produced by the
wave geometry.  Terms of that kind obtained by solving the inhomogeneous
mode equations arise from momentum-singular homogeneous data and are
excluded by matching and Hadamard admissibility.
\end{proposition}

\begin{proof}
For the first profile, the characteristic equations
\eqref{eq:g-general}--\eqref{eq:h-general} show that every perturbative
solution is the sum of a fixed particular solution and freely specifiable
functions of the characteristic invariants.  Exact-mode matching fixes the
nonconstant part of their sum.  Substitution of the matched solution into
the spectral representation gives Eq.~\eqref{eq:GF-first-order}, while the
direct expansions of the world function and Van Vleck determinant give the
same result through second order in Eq.~\eqref{eq:GF-second-order}.

For the second profile, the corrected metric has a perturbation linear in
$a$.  Its characteristic solution
\eqref{eq:second-profile-matched-correction} contains only nonnegative
powers of the transverse momenta.  In the spectral sum these become local
transverse derivatives of the flat propagator.  Equivalently, the first
variation of Synge's function gives
Eq.~\eqref{eq:second-profile-GF-first-order}.  Both derivations therefore
produce only derivatives of the ordinary Hadamard pole.

The remaining freedom in either calculation is a solution of the
homogeneous field equation.  A momentum-space choice containing inverse
powers of $q\bar q$ Fourier transforms into the transverse logarithms
identified in Sec.~\ref{sec:origin-logs}.  It is not fixed by the source,
does not follow from the geometry, and changes the singular part of the
two-point function.  Since two Hadamard two-point functions may differ only
by a smooth bisolution, such a choice is inadmissible for the matched
Hadamard state.  This proves Eq.~\eqref{eq:main-proposition} to the stated
orders.
\end{proof}

The proposition is deliberately local and perturbative.  Statements
directly on $U=0$ require a prescription for products of distributions.
The smooth-front families constructed above provide independent tests of
that limit; they are not needed for the off-surface conclusion proved here.

\subsection{Discussion and implications}
\label{sec:discussion}

The calculation separates an ambiguity of representation from an ambiguity
of state.  Splitting one inhomogeneous
equation into holomorphic and antiholomorphic parts is only a bookkeeping
choice.  If the corresponding homogeneous functions are transformed at the
same time so that $g+h$ is unchanged, the mode and its Green function are
unchanged.  A different singularity appears only when the homogeneous
completion is selected independently, in which case the incoming data have
also been changed.  This is the concrete realization of
Theorem~\ref{thm:general-diagnostic}.

A logarithm produced after dimensional regularization or introduction of a
Schwinger parameter need not be geometrical.  In the present calculation the
regularization evaluates a
Fourier transform already made singular by an inverse power of $q\bar q$; it
does not create that inverse power.  Local curvature information enters
instead through Synge's world function and the Van Vleck determinant.  Their
independent expansion reproduces the matched spectral result and provides a
state-independent diagnostic of its singular part.

In the first example, the exact solution fixes the perturbative branch and
shows how matching removes the apparently arbitrary characteristic
functions.  In the second, the metric specialization must first be corrected.  Its
transverse cross term and leading scalar operator are linear in $a$.  Once
this order counting is restored, the source contains positive
transverse-momentum powers and the earlier logarithmic mechanism is absent.
The mollified second-profile calculation reaches the same matched mode and
world-function correction, so this conclusion is not an artifact of having
expanded directly at a discontinuous front.

The vanishing of the renormalized stress tensor agrees with the standard
plane-wave result.  Here the admissibility condition is first imposed on the
two-point function, after which we subtract its
universal geometrical singularity, and only then evaluate the smooth
state-dependent remainder.  This ordering matters because the Hadamard
condition is required for the local definition of composite observables, not
merely for the numerical value of one particular observable.

The calculation has two limitations.  Our detailed perturbative checks
extend through second order for the first profile and first order for the
corrected second profile.  Section~\ref{sec:smooth-front} proves smeared
convergence for the first canonical smooth-front family, and
Sec.~\ref{sec:second-smooth-front} supplies the corresponding first-order
check for the second profile, in each case for the $R\neq0$ sector on compact
normal neighbourhoods.  These results do not define composite observables with both
arguments constrained exactly to the null surface, nor do they fix the
characteristic $R=0$ infrared sector.  Those questions require an additional
distributional prescription \cite{Steinbauer1998,Samann2024} and are separate
from the local Hadamard conclusion established here.

\subsection{Conclusions}
\label{sec:conclusions}

For a fixed normally hyperbolic operator, altering the split of an
inhomogeneous source changes a two-point function only through homogeneous
data.  A nonsmooth change cannot describe the same matched Hadamard state.

Applications to two impulsive plane-wave profiles demonstrate the criterion.
For the first profile, the corrected exact mode, characteristic solution, and
world-function calculation agree.  For the second, direct substitution into
the generating metric changes both its coefficients and its leading
perturbative order.  Neither corrected problem produces an additional
logarithmic short-distance singularity.  The logarithms instead encode
momentum-singular homogeneous data, while the admissible matched propagator
has the standard geometrical Hadamard form.  For the exactly soluble first
profile, the standard flat-normalized prescription also gives a vanishing
renormalized stress tensor away from the impulse.  The smooth-front family confirms that
the impulsive limit preserves this singularity class and does not generate
the inverse transverse-momentum factors needed for the spurious logarithm.

The same test applies whenever a perturbative mode equation has nonunique
particular solutions: verify the operator, retain the homogeneous completion,
impose common matching data, and compare the result with the Hadamard
parametrix.  The two plane-wave calculations give explicit examples of this
procedure.

\section*{Acknowledgments}
This paper is dedicated to the memory of Jan K\v{r}\'{\i}\v{z}, a dear friend and colleague who left us far too soon.


\section*{Competing interests}
The author declares no competing interests.

\section*{Data availability statement}
\label{sec:data-availability}
No new data were created or analysed in this study.

\appendix
\section{Audit of the 1999 calculation}
\label{sec:comparison}

Table~\ref{tab:first-profile-comparison} records the audit of the first
profile.  The purpose is not merely to list typographical differences.
It separates three logically distinct ingredients which were partly
conflated in the earlier calculation: the differential operator fixed by
the metric, a particular solution of the perturbative mode equation, and
the homogeneous data which specify the state.  Equation numbers in the
``1999'' column refer to the published article.

\begin{table*}[t]
\caption{First profile: comparison of the published and corrected calculations.}
\label{tab:first-profile-comparison}
\begingroup
\small
\sloppy
\renewcommand{\arraystretch}{1.18}
\begin{tblr}{width=\textwidth, colspec={Q[l,wd=2.1cm] Q[l,wd=2.7cm] X[l] X[l]}, cells={valign=t}, row{1}={font=\bfseries}, rowsep=2pt}
\toprule
1999 & Published step & Corrected result & Consequence \\
\midrule
Eq.~(4) & The exact d'Alembertian is written in real transverse coordinates. & Transforming consistently with $z=(x+\ii y)/\sqrt2$ gives the operator used in Eq.~\eqref{eq:exact-box}. & This fixes all numerical factors in the separated phase and supplies a direct check on the perturbative operator. \\
Eq.~(6) & The two transverse phase terms appear without the factor implied by the exact transverse integrals. & The exact mode is Eq.~\eqref{eq:exact-mode}; its phase denominators contain the required factor $1/4$. & Expansion of the corrected exact mode now agrees with the solution of the corrected first-order equation. \\
Eqs.~(19)--(21) & The inhomogeneous term is proportional to $(k_1^2+k_2^2)U$ and is divided between the holomorphic and antiholomorphic equations. & Direct action of $L_1=2U(\partial_x^2-\partial_y^2)$ gives $(k_1^2-k_2^2)U$. The resulting master equation is Eq.~\eqref{eq:correct-first-order-PDE}. & The source has spin-two transverse structure; replacing the difference by a sum changes the field equation, not merely its parametrization. \\
Eqs.~(22), (23) & A chosen split produces inverse powers of $k_1\mathbin{\pm}\ii k_2$. & The characteristic solution Eqs.~\eqref{eq:g-general}--\eqref{eq:h-general} contains arbitrary homogeneous functions; exact matching fixes their nonconstant sum. & The apparently different ``groupings'' are particular solutions supplemented by different free spectral data. \\
Eqs.~(28)--(31) & Momentum denominators are Schwinger-regularized and a Hankel function is interpreted as an extra logarithmic short-distance singularity. & The relevant two-dimensional transform is Eq.~\eqref{eq:two-dimensional-transform}; its $K_0$ logarithm originates from an inverse transverse momentum factor in the homogeneous data. & The logarithm is not generated by the local geometry. It changes the Hadamard singular part and is inadmissible for the matched state. \\
Eq.~(39) & A less singular split is said to recover the same coincidence degree as the exact answer. & The complete matched result is the expansion of $\Delta^{1/2}/(2\sigma-\ii0)$, explicitly Eqs.~\eqref{eq:GF-first-order} and \eqref{eq:GF-second-order}. & Agreement is strengthened from power counting at coincidence to equality with the geometrical Hadamard coefficients through the calculated orders. \\
\bottomrule
\end{tblr}
\endgroup
\end{table*}

The second-profile discrepancies begin one step earlier, at the metric
itself, and consequently propagate through the remainder of that
calculation.  They are summarized in
Table~\ref{tab:second-profile-comparison}.

\begin{table*}[t]
\caption{Second profile and final physical interpretation.}
\label{tab:second-profile-comparison}
\begingroup
\small
\sloppy
\renewcommand{\arraystretch}{1.18}
\begin{tblr}{width=\textwidth, colspec={Q[l,wd=2.1cm] Q[l,wd=2.7cm] X[l] X[l]}, cells={valign=t}, row{1}={font=\bfseries}, rowsep=2pt}
\toprule
1999 & Published step & Corrected result & Consequence \\
\midrule
Eq.~(40) & Substitution of $G(z)=az^2/2$ gives an expression with misplaced $z,\bar z$, an incorrect sign and coefficient, and extra powers of $a$. & Direct substitution into Hogan's metric gives Eq.~\eqref{eq:correct-second-metric}, with $B=\cos^2(aU)-z\bar z\sin^2(aU)$. & The published second metric is not the stated specialization of the generating metric. \\
Eq.~(41) & The wave operator is calculated from the published form of Eq.~(40). & Inverting the corrected transverse block and including its determinant gives Eq.~\eqref{eq:second-operator-expansion}. & Every subsequent source term must be recomputed; the old perturbative equations cannot be repaired by changing an integration constant. \\
Eq.~(42) & The first nontrivial correction is taken to be $O(a^2)$. & The off-diagonal transverse metric perturbation is already linear in $a$: $L_1=2U(\bar z\partial_z^2+z\partial_{\bar z}^2)$, as in Eq.~\eqref{eq:second-L1}. & The perturbative order and tensor structure of the old calculation are both incorrect. \\
Eqs.~(44)--(53) & Alternative splittings yield expressions with or without inverse transverse momenta and hence with or without logarithms. & The correct source is $L_1\phi^{(0)}=-U(\bar z\bar q^2+zq^2)\phi^{(0)}$; its matched solution is Eq.~\eqref{eq:second-profile-matched-correction}. & It contains no inverse $q$ or $\bar q$ poles. Its Green function, Eq.~\eqref{eq:second-profile-GF-first-order}, has only the first variation of the Hadamard pole. \\
Sec.~4 & The extra terms are called spurious because they leave no finite remainder after differentiating and taking coincidence; the vanishing of $\langle T_{\mu\nu}\rangle_{\rm ren}$ is retained. & The stronger reason is that momentum-singular additions fail the Hadamard-state comparison test. For the standard flat-normalized in-state, the exact plane-wave propagator gives Eq.~\eqref{eq:Tmunu-zero}. & The earlier physical conclusion is preserved, but it is now obtained after excluding the inadmissible two-point functions rather than by cancelling their effects only at the final observable. \\
\bottomrule
\end{tblr}
\endgroup
\end{table*}
\clearpage

This comparison changes the interpretation of the old ambiguity.  The
freedom to redistribute a common source between $g$ and $h$ is harmless
when the same homogeneous data are retained, because only $g+h$ enters the
mode.  A change of singularity structure occurs only when the
redistribution is accompanied by a change of homogeneous spectral data.
Thus the central intuition of the earlier work---that the additional
logarithms are spurious---survives, while their mathematical origin is
identified more precisely.



\begin{thebibliography}{99}
\sloppy

\bibitem{Fulling1978}
S.~A.~Fulling, M.~Sweeny, and R.~M.~Wald,
``Singularity structure of the two-point function in quantum field theory
in curved spacetime,'' Commun. Math. Phys. \textbf{63}, 257--264 (1978).
\href{https://doi.org/10.1007/BF01196934}{doi:10.1007/BF01196934}.

\bibitem{Radzikowski1996}
M.~J.~Radzikowski,
``Micro-local approach to the Hadamard condition in quantum field theory on
curved space-time,'' Commun. Math. Phys. \textbf{179}, 529--553 (1996).
\href{https://doi.org/10.1007/BF02100096}{doi:10.1007/BF02100096}.

\bibitem{KayWald1991}
B.~S.~Kay and R.~M.~Wald,
``Theorems on the uniqueness and thermal properties of stationary,
nonsingular, quasifree states on spacetimes with a bifurcate Killing
horizon,'' Phys. Rep. \textbf{207}, 49--136 (1991).
\href{https://doi.org/10.1016/0370-1573(91)90015-E}
{doi:10.1016/0370-1573(91)90015-E}.

\bibitem{Brunetti1996}
R.~Brunetti, K.~Fredenhagen, and M.~K\"ohler,
``The microlocal spectrum condition and Wick polynomials of free fields on
curved spacetimes,'' Commun. Math. Phys. \textbf{180}, 633--652 (1996).
\href{https://doi.org/10.1007/BF02099626}{doi:10.1007/BF02099626}.

\bibitem{HollandsWald2001}
S.~Hollands and R.~M.~Wald,
``Local Wick polynomials and time ordered products of quantum fields in
curved spacetime,'' Commun. Math. Phys. \textbf{223}, 289--326 (2001).
\href{https://doi.org/10.1007/s002200100540}
{doi:10.1007/s002200100540}.

\bibitem{DeWittBrehme1960}
B.~S.~DeWitt and R.~W.~Brehme,
``Radiation damping in a gravitational field,'' Ann. Phys. \textbf{9},
220--259 (1960).
\href{https://doi.org/10.1016/0003-4916(60)90030-0}
{doi:10.1016/0003-4916(60)90030-0}.

\bibitem{Christensen1976}
S.~M.~Christensen,
``Vacuum expectation value of the stress tensor in an arbitrary curved
background: The covariant point-separation method,'' Phys. Rev. D
\textbf{14}, 2490--2501 (1976).
\href{https://doi.org/10.1103/PhysRevD.14.2490}
{doi:10.1103/PhysRevD.14.2490}.

\bibitem{DecaniniFolacci2006}
Y.~D\'ecanini and A.~Folacci,
``Hadamard renormalization of the stress-energy tensor for a quantized
scalar field in a general spacetime of arbitrary dimension,'' Phys. Rev. D
\textbf{73}, 044027 (2006).
\href{https://doi.org/10.1103/PhysRevD.73.044027}
{doi:10.1103/PhysRevD.73.044027}.

\bibitem{BirrellDavies1982}
N.~D.~Birrell and P.~C.~W.~Davies,
\emph{Quantum Fields in Curved Space} (Cambridge University Press,
Cambridge, 1982).
\href{https://doi.org/10.1017/CBO9780511622632}
{doi:10.1017/CBO9780511622632}.

\bibitem{WaldBook1994}
R.~M.~Wald,
\emph{Quantum Field Theory in Curved Spacetime and Black Hole
Thermodynamics} (University of Chicago Press, Chicago, 1994).
\href{https://doi.org/10.7208/chicago/9780226870274.001.0001}
{\nolinkurl{doi:10.7208/chicago/9780226870274.001.0001}}.

\bibitem{PoissonPoundVega2011}
E.~Poisson, A.~Pound, and I.~Vega,
``The motion of point particles in curved spacetime,'' Living Rev. Relativ.
\textbf{14}, 7 (2011).
\href{https://doi.org/10.12942/lrr-2011-7}{doi:10.12942/lrr-2011-7}.

\bibitem{Gibbons1975}
G.~W.~Gibbons,
``Quantized fields propagating in plane-wave spacetimes,''
Commun. Math. Phys. \textbf{45}, 191--202 (1975).
\href{https://doi.org/10.1007/BF01629249}{doi:10.1007/BF01629249}.

\bibitem{Deser1975}
S.~Deser,
``Plane waves do not polarize the vacuum,'' Phys. Rev. Lett. \textbf{35},
1385--1387 (1975).
\href{https://doi.org/10.1103/PhysRevLett.35.1385}
{doi:10.1103/PhysRevLett.35.1385}.

\bibitem{Chrzanowski1976}
P.~L.~Chrzanowski, R.~A.~Matzner, V.~D.~Sandberg, and M.~P.~Ryan, Jr.,
``Zero-mass plane waves in nonzero gravitational backgrounds,''
Phys. Rev. D \textbf{14}, 317--326 (1976).
\href{https://doi.org/10.1103/PhysRevD.14.317}
{doi:10.1103/PhysRevD.14.317}.

\bibitem{Harte2013}
A.~I.~Harte,
``Tails of plane wave spacetimes: Wave-wave scattering in general
relativity,'' Phys. Rev. D \textbf{88}, 084059 (2013).
\href{https://doi.org/10.1103/PhysRevD.88.084059}
{doi:10.1103/PhysRevD.88.084059}.

\bibitem{Cho2023}
H.-T.~Cho,
``Wightman function and stochastic gravity noise kernel in impulsive
plane wave spacetimes,'' Phys. Rev. D \textbf{108}, 105007 (2023).
\href{https://doi.org/10.1103/PhysRevD.108.105007}
{doi:10.1103/PhysRevD.108.105007}.

\bibitem{Wald1977}
R.~M.~Wald,
``The back reaction effect in particle creation in curved spacetime,''
Commun. Math. Phys. \textbf{54}, 1--19 (1977).
\href{https://doi.org/10.1007/BF01609833}{doi:10.1007/BF01609833}.

\bibitem{Steinbauer1998}
R.~Steinbauer,
``Geodesics and geodesic deviation for impulsive gravitational waves,''
J. Math. Phys. \textbf{39}, 2201--2212 (1998).
\href{https://doi.org/10.1063/1.532283}{doi:10.1063/1.532283}.

\bibitem{KunzingerSteinbauer1999}
M.~Kunzinger and R.~Steinbauer,
``A rigorous solution concept for geodesic and geodesic deviation
equations in impulsive gravitational waves,'' J. Math. Phys. \textbf{40},
1479--1489 (1999).
\href{https://doi.org/10.1063/1.532809}{doi:10.1063/1.532809}.

\bibitem{SteinbauerVickers2006}
R.~Steinbauer and J.~A.~Vickers,
``The use of generalised functions and distributions in general
relativity,'' Class. Quantum Grav. \textbf{23}, R91--R114 (2006).
\href{https://doi.org/10.1088/0264-9381/23/10/R01}
{doi:10.1088/0264-9381/23/10/R01}.

\bibitem{PodolskyGriffiths1999}
J.~Podolsk\'y and J.~B.~Griffiths,
``Non-expanding impulsive waves in spacetimes of constant curvature,''
Class. Quantum Grav. \textbf{16}, 2937--2946 (1999).
\href{https://doi.org/10.1088/0264-9381/16/9/311}
{doi:10.1088/0264-9381/16/9/311}.

\bibitem{PodolskySteinbauer2022}
J.~Podolsk\'y and R.~Steinbauer,
``Penrose junction conditions with $\Lambda$: Geometric insights into
low-regularity metrics for impulsive gravitational waves,'' Gen. Relativ.
Gravit. \textbf{54}, 96 (2022).
\href{https://doi.org/10.1007/s10714-022-02977-6}
{doi:10.1007/s10714-022-02977-6}.

\bibitem{Samann2024}
C.~S\"amann, B.~Schinnerl, R.~Steinbauer, and R.~\v{S}varc,
``Cut-and-paste for impulsive gravitational waves with $\Lambda$:
the mathematical analysis,'' Lett. Math. Phys. \textbf{114}, 58 (2024).
\href{https://doi.org/10.1007/s11005-024-01804-0}
{doi:10.1007/s11005-024-01804-0}.

\bibitem{HortacsuLutfuoglu1999}
M.~Horta\c{c}su and B.~C.~L\"utf\"uo\u{g}lu,
``Spurious effects in perturbative calculations,''
Gen. Relativ. Gravit. \textbf{31}, 391--403 (1999),
\href{https://arxiv.org/abs/gr-qc/9812037}{arXiv:gr-qc/9812037}.

\bibitem{Hogan1999}
P.~A.~Hogan,
``Plane gravitational waves and Bertotti--Robinson space-times,''
Math. Proc. R. Ir. Acad. A \textbf{99}, 51--55 (1999).

\fussy
\end{thebibliography}
\end{document}